\def\llncs{0}
\def\fullpage{1}
\def\anonymous{0}
\def\authnote{1}
\def\notxfont{0}
\def\submission{0}

\ifnum\submission=1
\def\llncs{1}
\fi

\ifnum\llncs=1
\documentclass[envcountsect,a4paper,runningheads,10pt]{llncs}
 \usepackage[margin=1in]{geometry}
\else
	\documentclass[letterpaper,hmargin=1.05in,vmargin=1.05in,11pt]{article}
			\ifnum\fullpage=1
		\usepackage{fullpage}
		\fi
\fi

\usepackage[%
  colorlinks=true,
  citecolor=blue,
  pagebackref=true
]{hyperref}

\usepackage{amsmath, amsfonts, amssymb, mathtools,amscd}

\usepackage{amsthm}

\usepackage{lmodern}
\usepackage[T1]{fontenc}
\usepackage[utf8]{inputenc}

\usepackage{etex}

\usepackage{arydshln} 
\usepackage{url}
\usepackage{ifthen}
\usepackage{bm}
\usepackage{multirow}
\usepackage[dvips]{graphicx}
\usepackage[usenames]{color}
\usepackage{xcolor,colortbl} 
\usepackage{threeparttable}
\usepackage{comment}
\usepackage{paralist,verbatim}
\usepackage{cases}
\usepackage{booktabs}
\usepackage{braket}
\usepackage{cancel} 
\usepackage{ascmac} 
\usepackage{framed}
\usepackage{authblk}
\usepackage{pifont}
\usepackage{qcircuit}
\definecolor{darkblue}{rgb}{0,0,0.6}
\definecolor{darkgreen}{rgb}{0,0.5,0}
\definecolor{maroon}{rgb}{0.5,0.1,0.1}
\definecolor{dpurple}{rgb}{0.2,0,0.65}

\usepackage[capitalise,noabbrev]{cleveref}
\usepackage[absolute]{textpos}
\usepackage[final]{microtype}
\usepackage[absolute]{textpos}
\usepackage{everypage}
\DeclareMathAlphabet{\mathpzc}{OT1}{pzc}{m}{it}

\usepackage{algorithmic}
\usepackage{algorithm}
\usepackage{here}

\newtheoremstyle{thicktheorem}%
{\topsep}
{\topsep}
{\itshape}{}%
{\bfseries}%
{.}
{ }%
{\thmname{#1}\thmnumber{ #2}%
		\thmnote{ (#3)}%
}

\newtheoremstyle{remark}
{\topsep}
{\topsep}
	{}
	{}
	{}
	{.}
	{ }
	{\textit{\thmname{#1}}\thmnumber{ #2}
			\thmnote{ (#3)}%
	}

\ifnum\llncs=0
	\theoremstyle{thicktheorem}
	\newtheorem{theorem}{Theorem}[section]
	\newtheorem{lemma}[theorem]{Lemma}
	\newtheorem{corollary}[theorem]{Corollary}
	
	\newtheorem{definition}[theorem]{Definition}

	\theoremstyle{remark}
	
	\newtheorem{remark}[theorem]{Remark}

\else
\fi

\Crefname{MyClaim}{Claim}{Claims}

	\crefname{theorem}{Theorem}{Theorems}
	\crefname{assumption}{Assumption}{Assumptions}
	\crefname{construction}{Construction}{Constructions}
	\crefname{corollary}{Corollary}{Corollaries}
	\crefname{conjecture}{Conjecture}{Conjectures}
	\crefname{definition}{Definition}{Definitions}
	\crefname{exmaple}{Example}{Examples}
	\crefname{experiment}{Experiment}{Experiments}
	\crefname{counterexample}{Counterexample}{Counterexamples}
	\crefname{lemma}{Lemma}{Lemmata}
	\crefname{observation}{Observation}{Observations}
	\crefname{proposition}{Proposition}{Propositions}
	\crefname{remark}{Remark}{Remarks}
	\crefname{claim}{Claim}{Claims}
	\crefname{fact}{Fact}{Facts}
	\crefname{note}{Note}{Notes}

\ifnum\llncs=1
 \crefname{appendix}{App.}{Appendices}
 \crefname{section}{Sec.}{Sections}
\else
\fi

\ifnum\llncs=1
\renewcommand*{\backref}[1]{}
\else
	\renewcommand*{\backref}[1]{(Cited on page~#1.)}
	\ifnum\notxfont=1
	\else
		\usepackage{newtxtext}
	\fi
\fi
\ifnum\authnote=0  
\newcommand{\mor}[1]{}
\newcommand{\ryo}[1]{}
\newcommand{\takashi}[1]{}
\newcommand{\fuyuki}[1]{}

\else
\newcommand{\mor}[1]{$\ll$\textsf{\color{red} Tomoyuki: { #1}}$\gg$}
\newcommand{\takashi}[1]{$\ll$\textsf{\color{orange} Takashi: { #1}}$\gg$}
\newcommand{\ryo}[1]{$\ll$\textsf{\color{darkgreen} Ryo: { #1}}$\gg$}
\newcommand{\fuyuki}[1]{$\ll$\textsf{\color{darkblue} Fuyuki: { #1}}$\gg$}
\newcommand{\alex}[1]{$\ll$\textsf{\color{purple} Alex: { #1}}$\gg$}

\fi

\newcommand{\proj}[1]{\ket{#1}\!\bra{#1}}
\newcommand{\calY}{\mathcal{Y}}
\newcommand{\calX}{\mathcal{X}}
\newcommand{\swap}{\mathsf{Swap}}
\newcommand{\eval}{\mathsf{Eval}}

\newcommand{\sfot}{\mathsf{ot}}

\newcommand{\Tr}{\mathrm{Tr}}

\newcommand{\StateGen}{\mathsf{StateGen}}

\newcommand{\cert}{\keys{cert}}

\newcommand{\FullVer}{\mathsf{FullVer}}
\newcommand{\SemiVer}{\mathsf{SemiVer}}

\newcommand{\seteq}{\coloneqq}

\newcommand{\concat}{\|}

\newcommand{\cA}{\mathcal{A}}
\newcommand{\cB}{\mathcal{B}}
\newcommand{\cC}{\mathcal{C}}
\newcommand{\cD}{\mathcal{D}}

\newcommand{\cO}{\mathcal{O}}

\newcommand{\cX}{\mathcal{X}}
\newcommand{\cY}{\mathcal{Y}}
\newcommand{\cZ}{\mathcal{Z}}

\def\makeuppercase#1{
\expandafter\newcommand\csname tl#1\endcsname{\widetilde{#1}}
}

\def\makelowercase#1{
\expandafter\newcommand\csname tl#1\endcsname{\widetilde{#1}}
}

\newcommand{\N}{\mathbb{N}}

\newcommand{\regF}{\mathbf{F}}

\newcommand{\regB}{\mathbf{B}}
\newcommand{\regA}{\mathbf{A}}

\newcommand{\secp}{\lambda}

\newcommand{\A}{\entity{A}}

\newcommand*{\sk}{\keys{sk}}
\newcommand*{\pk}{\keys{pk}}

\newcommand*{\sigk}{\keys{sigk}}
\newcommand*{\snum}{\keys{snum}}

\newcommand*{\ck}{\keys{ck}}

\newcommand{\ct}{\keys{ct}}

\newcommand*{\pp}{\keys{pp}}

\newcommand*{\vk}{\keys{vk}}

\newcommand*{\td}{\keys{td}}

\newcommand*{\msg}{\keys{msg}}

\newcommand*{\keys}[1]{\mathsf{#1}}

\newcommand*{\algo}[1]{\ensuremath{\mathsf{#1}}}

\newcommand*{\entity}[1]{\mathcal{#1}}

\newenvironment{boxfig}[2]{\begin{figure}[#1]\fbox{\begin{minipage}{0.97\linewidth}
                        \vspace{0.2em}
                        \makebox[0.025\linewidth]{}
                        \begin{minipage}{0.95\linewidth}
            {{
                        #2 }}
                        \end{minipage}
                        \vspace{0.2em}
                        \end{minipage}}}{\end{figure}}

\newcommand{\pprotocol}[4]{
\begin{boxfig}{h}{\footnotesize 
\centering{\textbf{#1}}
    #4
\vspace{0.2em} } \caption{\label{#3} #2}
\end{boxfig}
}

\newcommand{\protocol}[4]{
\pprotocol{#1}{#2}{#3}{#4} }

\newcommand{\bit}{\{0,1\}}

\newcommand{\Setup}{\algo{Setup}}
\newcommand{\setup}{\algo{Setup}}
\newcommand{\Gen}{\algo{Gen}}
\newcommand{\KeyGen}{\algo{KeyGen}}

\newcommand{\Enc}{\algo{Enc}}

\newcommand{\Dec}{\algo{Dec}}

\newcommand{\Sign}{\algo{Sign}}

\newcommand{\Ver}{\algo{Ver}}
\newcommand{\Del}{\algo{Del}}
\newcommand{\Cert}{\algo{Cert}}

\newcommand{\st}{\algo{st}}

\newcommand{\PRF}{\algo{PRF}}

\newcommand{\negl}{{\mathsf{negl}}}

\newcommand{\poly}{{\mathrm{poly}}}

\makeatletter
\DeclareRobustCommand
  \myvdots{\vbox{\baselineskip4\p@ \lineskiplimit\z@
    \hbox{.}\hbox{.}\hbox{.}}}
\makeatother

\title{
Revocable Quantum Digital Signatures
}

\ifnum\anonymous=1
\ifnum\llncs=1
\author{\empty}\institute{\empty}
\else
\author{}
\fi
\else
\ifnum\llncs=1
\author{
Tomoyuki Morimae\inst{1} \and Takashi Yamakawa\inst{1,2} 
}
\institute{
	Yukawa Institute for Theoretical Physics, Kyoto University, Kyoto, Japan \and NTT Social Informatics Laboratories, Tokyo, Japan
}
\else
\author[1]{Tomoyuki Morimae}
\author[2,3]{Alexander Poremba}
\author[1,4,5]{Takashi Yamakawa}
\affil[1]{{\small Yukawa Institute for Theoretical Physics, Kyoto University, Kyoto, Japan}\authorcr{\small tomoyuki.morimae@yukawa.kyoto-u.ac.jp}}
\affil[2]{{\small 
Computing and Mathematical Sciences, Caltech, Pasadena, CA, USA}}
\affil[3]{{\small 
CSAIL and Department of Mathematics, MIT, Cambridge, MA, USA}\authorcr{\small poremba@mit.edu}}
\affil[4]{{\small NTT Social Informatics Laboratories, Tokyo, Japan}\authorcr{\small 
takashi.yamakawa@ntt.com}}
\affil[5]{{\small NTT Research Center for Theoretical Quantum Information, Atsugi, Japan}}

\fi 
\fi

\date{}

\begin{document}

\maketitle

\begin{abstract}
We study digital signatures with \emph{revocation capabilities} and show two results.
First, we define and construct digital signatures with revocable signing keys from the LWE assumption.
In this primitive, the signing key is a quantum state which enables a user to sign many messages and yet,
the quantum key is also \emph{revocable}, i.e., it can be collapsed into a classical certificate
which can later be verified. Once the key is successfully revoked,
we require that the initial recipient of the key loses the ability to sign.
We construct digital signatures with revocable signing keys
from a newly introduced primitive which we call \emph{two-tier one-shot signatures}, which may be of independent interest.
This is a variant of one-shot signatures, where
the verification of a signature for the message ``0'' is done publicly, whereas the verification for the message ``1'' is done in private.
We give a construction of two-tier one-shot signatures from the LWE assumption. As a complementary result, we also construct digital signatures with {\it quantum} revocation from group actions, where the quantum signing key is simply ``returned'' and then verified as part of revocation.

Second, we define and construct digital signatures with revocable signatures from OWFs.
In this primitive, the signer can produce quantum signatures which can later be revoked.
Here, the security property requires that, once revocation is successful,
the initial recipient of
the signature loses the ability to find accepting inputs to the signature verification algorithm.
We construct this primitive using a newly introduced \emph{two-tier} variant of tokenized signatures.
For the construction, we show a new lemma which we call the adaptive hardcore bit property for OWFs, which may enable further applications.
\end{abstract}
\ifnum\submission=1
\else
\clearpage
\newpage
\setcounter{tocdepth}{2}
\tableofcontents
\fi
\newpage

\section{Introduction}
\label{sec:introduction}

\subsection{Background}
The exotic nature of quantum physics, such as quantum superposition, no-cloning, entanglement, and uncertainty relations,
enables many new cryptographic applications which are impossible in a classical world. 
These include quantum money~\cite{Wiesner83}, copy-protection~\cite{CCC:Aaronson09,C:ALLZZ21}, secure software leasing~\cite{EC:AnaLap21}, unclonable encryption~\cite{Got02,TQC:BroLor20}, certified deletion~\cite{TCC:BroIsl20}, and more. 
Here, a common approach is to encode information into a quantum state which prevents it from being copied by the no-cloning principle.  

Following this line of research, Ananth, Poremba, and Vaikuntanathan~\cite{cryptoeprint:2023/325} and Agrawal, Kitagawa, Nishimaki, Yamada, and Yamakawa~\cite{EC:AKNYY23} concurrently introduced the concept of key-revocable public key encryption (PKE),\footnote{Agrawal et al.~\cite{EC:AKNYY23} call it PKE with secure key leasing.} which realizes the following functionality: a decryption capability is delegated to a user in the form of a quantum decryption key in such a way that, once the key is returned, the user loses the ability to decrypt.
They constructed  key-revocable PKE schemes based on standard assumptions, namely quantum hardness of the learning with errors problem (LWE assumption)~\cite{cryptoeprint:2023/325} or even the mere existence of any PKE scheme~\cite{EC:AKNYY23}. 
They also extended the idea of revocable cryptography to pseudorandom functions~\cite{cryptoeprint:2023/325} and encryption with advanced functionality such as attribute-based encryption and functional encryption~\cite{EC:AKNYY23}. 
However, neither of these works extended the idea to \emph{digital signatures} despite their great importance in cryptography. 
This state of affairs raises the following question:

\begin{center}
\emph{Is it possible to construct digital signature schemes with revocation capabilities?}
\end{center}

The delegation of privileges is of central importance in cryptography, and the task of revoking privileges in the context of digital signatures and certificates, in particular, remains a fundamental challenge for cryptography~\cite{stubble,Riv98b}. 
One simple solution is to use a limited-time delegatable signature scheme, where a certified signing key is generated together with an expiration date. Note that this requires that the expiration date is known ahead of time and that the clocks be synchronized. Moreover, issuing new keys (for example, each day) could potentially also be costly.
Quantum digital signature schemes with revocation capabilities could potentially resolve these difficulties by leveraging the power of quantum information.

To illustrate the use of \emph{revocable} digital signature schemes, consider the following scenarios. Suppose that an employee at a company, say Alice, takes a temporary leave of absence and wishes to authorize her colleague, say Bob, to sign a few important documents on her behalf. One thing Alice can do is to simply use a (classical) digital signature scheme and to share her signing keys with Bob. While this naïve approach would certainly allow Bob to produce valid signatures while Alice is gone, it also means that Bob continues to have access to the signing keys---long after Alice's return. 
This is because the signing key of a digital signature scheme is \emph{classical}, and hence it can be copied at will. In particular, a malicious Bob could secretly sell Alice's signing key to a third party for a profit.
A digital signature scheme with \emph{revocable signing keys} can remedy this situation as it enables Alice to certify that Bob has lost access to the signing key once and for all. 

As a second example, consider the following scenario. Suppose that a company or a governmental organization wishes to grant a new employee certain access privileges throughout their employment; for example to various buildings or office spaces. 
One solution is to use an \emph{electronic} ID card through a mobile device, where a digital signature is used for identity management. Naturally, one would like to ensure that, once the employee's contract is terminated, their ID card is disabled in the system and no longer allows for further unauthorized access.
However, if the signature corresponding to the employee's ID is a digital object, it is conceivable that the owner
of the card manages to retain their ID card even after it is disabled.
This threat especially concerns scenarios in which the verification of an ID card is performed by a device which is not connected to the internet, or simply not updated frequently enough.
A digital signature scheme with revocable signatures can remedy this situation as it 
 enables \emph{revocable quantum} ID cards; in particular, it allows one to certify that the initial access privileges have been revoked once and for all.

\subsection{Our Results}
In this paper, we show the following two results on revocable digital signatures.

\paragraph{Revocable signing keys.}
First, we define digital signatures with revocable {\it signing keys} ($\mathsf{DSR\mbox{-}Key}$). 
In this primitive, a signing key is encoded in the form of a quantum state which enables the recipient to sign many messages. However, once the key is successfully revoked from a user,
they no longer have the ability to generate valid signatures.
Here, we consider \emph{classical revocation}, i.e., a classical certificate is issued once the user destroys the quantum signing key with an appropriate measurement. In addition, the verification of the revocation certificate takes place in private, which means that the verification requires a private key which should be kept secret.
We construct $\mathsf{DSR\mbox{-}Key}$ based solely on the quantum hardness of the LWE problem~\cite{Reg05}. 
We remark that our scheme is inherently \emph{stateful}, i.e., whenever a user generates a new signature, the user must update the singing key for the next invocation of the signing algorithm. Indeed, we believe that digital signatures with revocable signing keys must be inherently stateful since a user must keep the quantum signing key as a ``state'' for generating multiple signatures. 
An undesirable feature of our scheme is that the signing key and signature sizes grow with the number of signatures to be generated.  

As complementary result, we also consider $\mathsf{DSR\mbox{-}Key}$ with ${\it quantum}$ revocation.
In this primitive, not a classical deletion certificate but the quantum signing key itself
is returned for the revocation. 
We construct the primitive from group actions with the one-wayness property~\cite{TCC:JQSY19}. 
The existence of group actions with the one-wayness property is incomparable with the LWE assumption.

\paragraph{Revocable signatures.}
Second, we define digital signatures with revocable {\it signatures} ($\mathsf{DSR\mbox{-}Sign}$). 
In this primitive, signatures are encoded as quantum states which can later be revoked. The security property
guarantees that, once revocation is successful, the initial recipient of the signature loses the ability to pass the signature verification.
We construct digital signatures with revocable signatures based on the existence of (quantum-secure) one-way functions (OWFs).  
In our scheme, the revocation is classical and private, i.e., a user can issue a classical certificate of revocation, which is verified by using a private key. 

\subsection{Comparison with Existing Works}
To our knowledge, there is no prior work that studies digital signatures with quantum signatures. 
On the other hand, 
there are several existing works that study digital signatures with quantum signing keys.  
We review them and compare them with our $\mathsf{DSR\mbox{-}Key}$. 

\begin{itemize}
\item 
{\bf Tokenized signatures~\cite{BenDavidSattath,C:CLLZ21}.} 
In a tokenized signature scheme, the signing key corresponds to a quantum state which can be used to generate a signature on at most one message.
At first sight, the security notion seems to imply the desired security guarantee for $\mathsf{DSR\mbox{-}Key}$, 
since a signature for a dummy message may serve as the classical deletion certificate for the signing key.\footnote{Note, however, that tokenized signatures offer public verification of signatures,
whereas certifying revocation in our $\mathsf{DSR\mbox{-}Key}$ scheme takes place in private.}
However, the problem is that tokenized signatures do not achieve the correctness for $\mathsf{DSR\mbox{-}Key}$; namely,
in tokenized signatures, a user who receives a quantum signing key can generate only a single signature, whereas in $\mathsf{DSR\mbox{-}Key}$, we require that a user can generate arbitrarily many signatures before the signing key is revoked. Thus, tokenized signatures are not sufficient for achieving our goal. 
A similar problem exists for semi-quantum tokenized signatures~\cite{C:Shmueli22} and one-shot signatures~\cite{STOC:AGKZ20} as well. 
\item
{\bf Copy-protection for digital signatures~\cite{TCC:LLQZ22} (a.k.a. single-signer signatures~\cite{STOC:AGKZ20}.\footnote{Technically speaking, \cite{TCC:LLQZ22} and \cite{STOC:AGKZ20} require slightly different security definitions, but high level ideas are the same.})} 
In this primitive, a signing key corresponds to a quantum state which cannot be copied. More precisely, suppose that a user is given one copy of the signing key and tries to split it into two signing keys. The security property requires that at most one of these two signing keys is capable at generating a valid signature on a random message.  
Amos, Georgiou, Kiayias, and Zhandry~\cite{STOC:AGKZ20} constructed such a signature scheme based on one-shot signatures. However, the only known construction of one-shot signatures is relative to classical oracles, and there is no known construction without oracles. Liu, Liu, Qian, and Zhandry~\cite{TCC:LLQZ22} constructed it based on indistinguishability obfuscation (iO) and OWFs.   
Intuitively, copy-protection for digital signatures implies $\mathsf{DSR\mbox{-}Key}$, because checking whether a returned signing key succeeds at generating valid signatures on random messages can serve a means of verification for revocation.\footnote{While this sounds plausible, there is a subtlety regarding the security definitions. Indeed, we believe that the security of copy-protection for digital signatures~\cite{TCC:LLQZ22} or single-signer signatures~\cite{STOC:AGKZ20} does not readily imply our security definition in \Cref{def:dswrsk}, though they do seem to imply some weaker but reasonable variants of security. See also \Cref{rem:definition_revocable_signing_key}.} 
Compared with this approach, our construction has the advantage that it is based on the standard assumption (namely the LWE assumption), whereas they require the very strong assumption of iO or ideal oracles.  On the other hand, a disadvantage of our construction is that revocation requires private information, whereas theirs have the potential for public revocation. Another disadvantage is that the size of the signing key (and signatures) grows with the number of signatures, whereas this is kept constant in \cite{TCC:LLQZ22} (but not in ~\cite{STOC:AGKZ20}). 
\end{itemize}

\subsection{Technical Overview}
Here we give intuitive explanations of our constructions.

\paragraph{Construction of $\mathsf{DSR\mbox{-}Key}$.}
Our first scheme, $\mathsf{DSR\mbox{-}Key}$, is constructed using \emph{two-tier one-shot signatures} (2-OSS),
which is a new primitive which we introduce in this paper.\footnote{The term ``two-tier" is taken from \cite{TCC:KitNisYam21} where they define two-tier quantum lightning, which is a similar variant of quantum lightning~\cite{EC:Zhandry19b}.}
2-OSS are variants of one-shot signatures~\cite{STOC:AGKZ20} for single-bit messages.
The main difference with regard to one-shot signatures is that there
are two verification algorithms, and a signature for the message ``0'' is verified by a public verification algorithm,
whereas a signature for the massage ``1'' is verified by a {\it private} verification algorithm.
We believe that the notion of 2-OSS may be of independent interest.
Our construction of 2-OSS is conceptually similar to the construction of two-tier quantum lightning in \cite{TCC:KitNisYam21}, and can be based solely on the LWE assumption.

From 2-OSS, we then go on to construct $\mathsf{DSR\mbox{-}Key}$.
We first construct $\mathsf{DSR\mbox{-}Key}$ for single-bit messages from 2-OSS as follows.\footnote{The scheme can be extended to the one for multi-bit messages by using the
collision resistant hash functions.} 
The signing key $\sigk$ of $\mathsf{DSR\mbox{-}Key}$ consists of a pair $(\sigk_0,\sigk_1)$ of signing keys of a 2-OSS scheme.
To sign a single-bit message $m \in \bit$, the message ``0'' is signed with the signing algorithm of a 2-OSS scheme using
the signing key $\sigk_m$.
Because the signature on $m$ corresponds to a particular signature of ``0'' with respect to the 2-OSS scheme, it can be verified with the public verification algorithm of 2-OSS.
To delete the signing key, the message ``1'' is signed with the signing algorithm of the 2-OSS scheme by using
the signing key.
The signature for the message ``1'' corresponds to the revocation certificate, and it can be verified using the private verification algorithm of 2-OSS.

Our aforementioned construction readily implies a \emph{one-time version} of a $\mathsf{DSR\mbox{-}Key}$ scheme, namely, the correctness and security hold when the signing is used only once.
We then upgrade it to the many-time version by using a similar chain-based construction of single-signer signatures from one-shot signatures as in~\cite{STOC:AGKZ20}. 
That is, it works as follows. The signing key and verification key of the many-time scheme are those of the one-time scheme, respectively. We denote them by $(\sfot.\sigk_{0},\sfot.\vk_{0})$. 
When signing on the first message $m_1$, the signer first generates a new key pair $(\sfot.\sigk_{1},\sfot.\vk_{1})$ of the one-time scheme, uses $\sfot.\sigk_{0}$ to sign on the concatenation $m_1\concat \sfot.\vk_{1}$ of the message and the newly generated verification key to generate a signature $\sfot.\sigma_1$ of the one-time scheme.
Then it outputs $(m_1,\sfot.\vk_{1},\sfot.\sigma_1)$ as a signature of the many-time scheme.\footnote{We include $m_1$ in the signature for notational convenience even though this is redundant.}
Similarly, when signing on the $k$-th message $m_k$ for $k\ge 2$, the signer generates a new key pair $(\sfot.\sigk_{k},\sfot.\vk_{k})$ and uses $\sfot.\sigk_{k-1}$ to sign on $m_k\concat \sfot.\vk_{k}$ to generate a signature $\sfot.\sigma_k$. 
Then the signature of the many-time scheme consists of $\{m_i,\sfot.\vk_{i},\sfot.\sigma_i\}_{i\in [k]}$. 
The verification algorithm of the many-time scheme verifies $\sfot.\sigma_i$ for all $i\in [k]$ under the corresponding message and verification key, and accepts if all of these verification checks pass. 
To revoke a signing key, the signer generates revocation certificates for all of the signing keys of the one-time scheme which have previously been generated, and the verification of the revocation certificate simply verifies that all these revocation certificates are valid.\footnote{The ability to verify all previously generated signing keys (e.g., as part of a chain) may require secret \emph{trapdoor information}.}
It is easy to reduce security of the above many-time scheme to that of the one-time scheme.

\paragraph{Construction of $\mathsf{DSR\mbox{-}Sign}$.}
Our second scheme, $\mathsf{DSR\mbox{-}Sign}$, is constructed from what we call \emph{two-tier tokenized signatures} (2-TS),
which is a new primitive introduced in this paper.
2-TS are
variants of tokenized signatures~\cite{BenDavidSattath} for single-bit messages where two signature verification algorithms exist.
One verification algorithm is used to verify signatures for the message ``0'', and it uses the public key.
The other verification algorithm is used to verify signatures for the message ``1'', and it uses the {\it secret} key.

We construct 2-TS from OWFs by using a new lemma that we call the \emph{adaptive hardcore bit property for OWFs}, inspired by a similar notion which was shown for a family of noisy trapdoor claw-free functions by Brakerski et al.~\cite{JACM:BCMVV21}.
We believe that our lemma may be of independent interest, and enable further applications down the line.
The adaptive hardcore bit property for OWFs roughly states that given
$\ket{x_0}+(-1)^c\ket{x_1}$ and $(f(x_0),f(x_1))$, no QPT adversary can output
$(x,d)$ such that $f(x)\in\{f(x_0),f(x_1)\}$ and $d\cdot(x_0\oplus x_1)=c$, where
$f$ is a OWF, $x_0,x_1\gets\bit^\ell$, and $c\gets\bit$.\footnote{We actually need its amplified version, because 
in this case the adversary can win with probability $1/2$ by measuring the state to get $x_0$ or $x_1$, and randomly choosing $d$.}
The adaptive hardcore bit property for OWFs is shown by using a theorem which is implicit in a recent work~\cite{bartusek2023weakening}.

From the adaptive hardcore bit property for OWFs, we construct 2-TS as follows:
The quantum signing token is $\ket{x_0}+(-1)^c\ket{x_1}$
with random $x_0,x_1\gets\bit^\ell$ and $c\gets\bit$.\footnote{Again, we actually consider its amplified version so that the winning probability of the adversary is negligibly small.}
The public key is $(f(x_0),f(x_1))$, where $f$ is a OWF, and the secret key is $(x_0,x_1,c)$.
To sign the message ``0'', the token is measured in the computational basis to obtain either $x_0$ or $x_1$.
To sign the message ``1'', the token is measured in the Hadamard basis to obtain a string $d$ such that $d\cdot(x_0\oplus x_1)=c$.
The measurement result in the computational basis is then verified with the public key, whereas
the measurement result in the Hadamard basis is verified with the secret key. 
Due to the adaptive hardcore bit property for OWFs (formally shown in \Cref{thm:adaptive_hardcore_bit}), no QPT adversary can output both signatures at the same time.

Finally, we observe that $\mathsf{DSR\mbox{-}Sign}$ can be constructed from any 2-TS scheme
by considering the quantum signature of $\mathsf{DSR\mbox{-}Sign}$ as a quantum signing token of 2-TS.
To verify the quantum signature, we sign the message ``0'' by using the quantum token, and verify it. 
To delete the quantum signature, we sign the message ``1'' by using the quantum token.
The verification of the revocation certificate requires one to check whether the deletion certificate is a valid signature for message ``1'' or not.

\subsection{Related Works}
We have already explained relations between our results and existing works on digital signatures with
quantum signing keys. Here, we give a brief review on other related quantum cryptographic primitives.

\paragraph{Certified deletion and revocation.}
Unruh~\cite{JACM:Unruh15} first initiated the study of quantum revocable encryption. This allows the recipient of a quantum ciphertext to return the state, thereby losing all information about the encrypted message.
Quantum encryption with certified deletion~\cite{AC:HMNY21,ITCS:Poremba23,cryptoeprint:2022/1178,cryptoeprint:2023/236,cryptoeprint:2023/265,cryptoeprint:2023/370}, 
first introduced by
Broadbent and Islam~\cite{TCC:BroIsl20}, 
enables the deletion of quantum ciphertexts, whereby a
classical certificate is produced which can be verified.
In particular, \cite{cryptoeprint:2022/1178,cryptoeprint:2023/236,C:HMNY22} study the certified everlasting security
where the security is guaranteed even against unbounded adversary once a valid deletion certificate is issued.
\cite{cryptoeprint:2023/538} and \cite{bartusek2023weakening} recently showed a general conversion technique to convert the certified everlasting lemma
by Bartusek and Kurana~\cite{cryptoeprint:2022/1178} for the private verification to
the public one assuming only OWFs (or even weaker assumptions such as 
hard quantum planted problems for {\bf NP} or the one-way states generators~\cite{C:MorYam22}). 

The notion of certified deletion has also been used to revoke cryptographic keys~\cite{AC:KitNis22,EC:AKNYY23,cryptoeprint:2023/265,cryptoeprint:2023/325,cryptoeprint:2023/1640}. Here, a key is delegated to a user in the form of a quantum state which can later be revoked.
Once the key is destroyed and a valid certificate is issued, the functionality associated with the key is no longer available to the user.

Finally, we remark that the notion of revocation has also been considered in the context of more general programs. Ananth and La Placa~\cite{EC:AnaLap21} introduced the notion of secure software leasing. Here, the security guarantees that the functionality of a piece of quantum software is lost once it is returned and verified.

\paragraph{Copy-protection.}
Copy-protection, introduced by Aaronson~\cite{CCC:Aaronson09}, is a primitive  which allows one to encode a functionality into a quantum state in such a way that it cannot be cloned.
\cite{C:ALLZZ21} showed that any unlearnable functionality can be copy-protected with a classical oracle. \cite{cryptoeprint:2020/1194} constructed copy-protection schemes for (multi-bit) point functions as well as compute-and-compare programs in the quantum random oracle model.
\cite{C:CLLZ21} constructed unclonable decryption schemes from 
iO and compute-and-compare obfuscation for the class of unpredictable distributions,
which were previously constructed with classical oracle in \cite{cryptoeprint:2020/877}.
\cite{C:CLLZ21} also constructed a copy-protection scheme for pseudorandom functions
assuming iO, OWFs, and compute-and-compare obfuscation for the class
of unpredictable distributions. 
\cite{TCC:LLQZ22} constructed bounded collusion-resistant copy-protection for various functionalities (copy-protection of decryption, digital signatures and PRFs)
with iO and LWE.

\subsection{Organization}
Our first result, $\mathsf{DSR\mbox{-}Key}$, is given in \cref{sec:skcert}.
It is constructed from 2-OSS whose definition and construction from the LWE assumption are given in \cref{sec:2-OSS}.
We give a construction of $\mathsf{DSR}\mbox{-}\mathsf{Key}$ with quantum revocation from group actions in \cref{sec:GA}.

Our second result, $\mathsf{DSR\mbox{-}Sign}$, is given in \cref{sec:sigcert}.
It is constructed from 2-TS, which is defined and constructed in \cref{sec:2-TS}.
2-TS is constructed from OWFs via the adaptive hardcore bit property for OWFs, which is explained in \cref{sec:AHB}.

\section{Preliminaries}\label{sec:preliminaries}

\subsection{Basic Notation}
\label{sec:basic_notations}

We use the standard notations of quantum computing and cryptography.
We use $\secp$ as the security parameter.
For any set $S$, $x\gets S$ means that an element $x$ is sampled uniformly at random from the set $S$.
We write $\negl$ to mean a negligible function.
PPT stands for (classical) probabilistic polynomial-time and QPT stands for quantum polynomial-time.
For an algorithm $A$, $y\gets A(x)$ means that the algorithm $A$ outputs $y$ on input $x$.
For two bit strings $x$ and $y$, $x\|y$ means the concatenation of them.
For simplicity, we sometimes omit the normalization factor of a quantum state.
(For example, we write $\frac{1}{\sqrt{2}}(|x_0\rangle+|x_1\rangle)$ just as
$|x_0\rangle+|x_1\rangle$.)
$I\coloneqq|0\rangle\langle0|+|1\rangle\langle 1|$ is the two-dimensional identity operator.
For the notational simplicity, we sometimes write $I^{\otimes n}$ just as $I$ when
the dimension is clear from the context.

\paragraph{Densities and Distances.}\label{sec:densities-distances}

Let $\mathcal{X}$ be a finite domain. A density $f$ on $\mathcal{X}$ is a function $f: \mathcal{X} \rightarrow [0,1]$ such that $\sum_{x \in \mathcal{X}} f(x) =1$. We denote by $\mathcal{D}_{\mathcal{X}}$ the set of densities on $\mathcal{X}$. For any $f \in \mathcal{D}_{\mathcal{X}}$,  we let $\mathsf{Supp}(f) \coloneqq \{x \in \mathcal{X} \, : \, f(x) >0 \}$. Given two densities $f_0,f_1$ over $\mathcal{X}$, the Hellinger distance between $f_0$ and $f_1$ is defined by
$$
\mathsf{H}^2(f_0,f_1) \coloneqq 1 - \sum_{x \in \mathcal{X}} \sqrt{f_0(x) f_1(x)}.
$$
For two density matrices $\rho$ and $\sigma$, the trace distance is defined as
$$\|\rho-\sigma\|_{\rm tr} \coloneqq 
\frac{1}{2}\|\rho-\sigma\|_1 =
\frac{1}{2} \mathrm{Tr}\left[\sqrt{(\rho - \sigma)^2}\right],
$$
where $\|\cdot\|_1$ is the trace norm.
The following elementary lemma relates the Hellinger distance and the trace distance of superpositions.

\begin{lemma}\label{lem:hellinger}
Let $\mathcal{X}$ be a finite set, $f_0,f_1 \in \mathcal{D}_{\mathcal{X}}$ and
$$
\ket{\psi_b} \coloneqq \sum_{x \in \mathcal{X}} \sqrt{f_b(x)} \ket{x}
$$
for $b \in \bit$. It holds that 
$$
\left\| \proj{\psi_0} - \proj{\psi_1} \right\|_{\rm tr} = \sqrt{1- (1- \mathsf{H}^2(f_0,f_1))^2}.
$$
\end{lemma}

\begin{theorem}[Holevo-Helstrom, \cite{HOLEVO1973337,Helstrom1969QuantumDA}]\label{thm:holevo-hesltrom} Consider an experiment in which one of two quantum states, either $\rho$ or $\sigma$, is sent to a distinguisher with probability $1/2$. Then, any measurement which seeks to discriminate between $\rho$ and $\sigma$ has success probability $p_{\text{succ}}$ at most
$$
p_{\text{succ}} \leq \frac{1}{2} + \frac{1}{2} \|\rho-\sigma\|_{\rm tr}.
$$  
\end{theorem}

\subsection{Cryptography}
\begin{definition}[EUF-CMA Secure Digital Signatures]
An EUF-CMA secure digital signature scheme is a set $(\KeyGen,\Sign,\Ver)$ of QPT algorithms such that
\begin{itemize}
    \item 
    $\KeyGen(1^\secp)\to(\sigk,\vk):$
    It is a QPT algorithm that, on input the security parameter $\secp$, outputs a classical signing key $\sigk$ and a classical verification
    key $\vk$.
    \item 
    $\Sign(\sigk,m)\to\sigma:$
    It is a QPT algorithm that, on input $\sigk$ and a message $m$, outputs a classical signature $\sigma$.
    \item 
    $\Ver(\vk,\sigma,m)\to\top\bot:$
    It is a QPT algorithm that, on input $\vk$, $\sigma$, and $m$, outputs $\top/\bot$.
\end{itemize}
   We require the following two properties.

   \paragraph{Correctness:}
   For any message $m$,
   \begin{align}
    \Pr[\top\gets\Ver(\vk,\sigma,m):(\sigk,\vk)\gets\KeyGen(1^\secp),\sigma\gets\Sign(\sigk,m)]\ge1-\negl(\secp).   
   \end{align}
   
   \paragraph{EUF-CMA security:}
   For any QPT adversary $\cA$,
   \begin{align}
    \Pr\left[\top\gets\Ver(\vk,\sigma,m^*):
    \begin{array}{rr}
    (\sigk,\vk)\gets\KeyGen(1^\secp)\\
    (m^*,\sigma)\gets\cA^{\Sign(\sigk,\cdot)}(\vk)
    \end{array}
    \right]\le\negl(\secp), 
   \end{align}
   where $\cA$ is not allowed to query $m^*$ to the signing oracle.
\end{definition}

\subsection{Noisy Trapdoor Claw-Free Hash Function Family}\label{sec:ntcf}

We now recall the notion of noisy trapdoor claw-free (NTCF) hash function family introduced by~\cite{JACM:BCMVV21}.

\begin{definition}[NTCF Hash Function Family~\cite{JACM:BCMVV21}]\label{def:NTCF}
Let $\mathcal{X},\mathcal{Y}$ be finite sets, let $\mathcal{D}_{\mathcal{Y}}$ be the set of probability densities over $\mathcal{Y}$, and $\mathcal{K}_{\mathcal{F}}$ a finite set of keys. A family of functions
$$
\mathcal{F} := \{f_{k,b} \, : \, \mathcal{X} \rightarrow \mathcal{D}_{\mathcal{Y}} \}_{k \in \mathcal{K}_{\mathcal{F}},b \in \{0,1\}}
$$
is a NTCF family if the following properties hold:
\begin{description}
\item[Efficient Function Generation:] There exists a PPT algorithm $\mathsf{NTCF}.\Gen_{\mathcal{F}}$ which generates a key $k \in \mathcal{K}_{\mathcal{F}}$ and a trapdoor $\mathsf{td}$.

\item[Trapdoor Injective Pair:] For all keys $k \in \mathcal{K}_{\mathcal{F}}$, the following holds:
\begin{itemize}
    \item Trapdoor: For all $b \in \{0,1\}$ and $x \neq x' \in \mathcal{X}$, $\mathsf{Supp}(f_{k,b}(x)) \cap \mathsf{Supp}(f_{k,b}(x')) = \emptyset$. In addition, there exists an efficient deterministic algorithm $\mathsf{Inv}_{\mathcal{F}}$ such that for all $b \in \{0,1\}$, $x \in \mathcal{X}$ and $y \in \mathsf{Supp}(f_{k,b}(x))$, we have $\mathsf{Inv}(\mathsf{td},b,y)=x$.

    \item Injective pair: There exists a perfect matching relation $\mathcal{R}_k \subseteq \mathcal{X} \times \mathcal{X}$ such that $f_{k,0}(x_0) = f_{k,1}(x_1)$ if and only if $(x_0,x_1) \in \mathcal{R}_k$.
\end{itemize}

\item[Efficient Range Superposition:] For all keys $k \in \mathcal{K}_{\mathcal{F}}$ and $b \in \bit$, there exists a function $f'_{k,b}: \mathcal{X} \rightarrow \mathcal{D}_{\mathcal{Y}}$ such that the following holds:
\begin{itemize}
    \item For all $(x_0,x_1) \in \mathcal{R}_k$ and $y \in \mathsf{Supp}(f'_{k,b}(x_b))$, it holds that both $\mathsf{Inv}_{\mathcal{F}}(\mathsf{td},b,y)=x_b$ and $\mathsf{Inv}_{\mathcal{F}}(\mathsf{td},b \oplus 1,y) = x_{b \oplus 1}$.

    \item There exists an efficient deterministic procedure $\mathsf{Chk}_{\mathcal{F}}$ that takes as input $k \in \mathcal{K}_{\mathcal{F}},b \in \bit$, $x \in \mathcal{X}$ and $y \in \mathcal{Y}$ and outputs $1$ if $y \in \mathsf{Supp}(f'_{k,b}(x))$ and $0$ otherwise, This procedure does not need a trapdoor $\mathsf{td}$.

    \item For all $k \in \mathcal{K}_{\mathcal{F}}$ and $b \in \bit$,
    $$
    \underset{x \leftarrow \mathcal{X}}{\mathbb{E}}\left[\mathsf{H}^2(f_{k,b}(x),f'_{k,b}(x))\right] \leq \negl(\lambda).
    $$
    Here, $\mathsf{H}^2$ is the Hellinger distance (see \Cref{sec:densities-distances}). In addition, there exists a QPT algorithm $\mathsf{Samp}_{\mathcal{F}}$ that takes as input $k$ and $b \in \bit$ and prepae the quantum state
    $$
    \ket{\psi'} = \frac{1}{\sqrt{|\mathcal{X}|}} \sum_{x \in \mathcal{X}, y\in \mathcal{Y}}
\sqrt{(f'_{k,b}(x))(y)} \ket{x} \ket{y}.
    $$
This property and \Cref{lem:hellinger} immediately imply that
$$
\| \proj{\psi} - \proj{\psi'} \|_{\rm tr} \leq \negl(\lambda),
$$
where $\ket{\psi} = \frac{1}{\sqrt{|\mathcal{X}|}} \sum_{x \in \mathcal{X}, y\in \mathcal{Y}}
\sqrt{(f_{k,b}(x))(y)} \ket{x} \ket{y}.$
\end{itemize}

\item[Adaptive Hardcore Bit:] For all keys $k \in \mathcal{X}_{\mathcal{F}}$, the following holds. For some integer $w$ that is a polynomially bounded function of $\lambda$,
\begin{itemize}
    \item For all $b \in \bit$ and $x \in \mathcal{X}$, there exists a set $\mathcal{G}_{k,b,x} \subseteq \bit^w$ such that $\Pr_{d \leftarrow\bit^w}[d \notin \mathcal{G}_{k,b,x}] \leq \negl(\lambda)$. In addition, there exists a PPT algorithm that checks for membership in $\mathcal{G}_{k,b,x}$ given $k,b,x$ and $\mathsf{td}$.

    \item There exists an efficiently computable injection $J: \mathcal{X} \rightarrow \bit^w$ such that $J$ can be inverted efficiently on its range, and such that the following holds. Let
\begin{align*}
    H_k &:= \{(b,x_b,d,d \cdot (J(x_0) \oplus J(x_1))) \, | \, b \in \bit, (x_0,x_1) \in \mathcal{R}_k, d \in \mathcal{G}_{k,0,x_0} \cap  \mathcal{G}_{k,1,x_1}\},\\
    \overline{H}_k &:= \{(b,x_b, d,c) \, | \, (b,x,d,c\oplus 1) \in H_k\},
\end{align*}
then, for any QPT algorithm $\mathcal{A}$, it holds that
$$
\vline \,\underset{(k,\mathsf{td}) \leftarrow \mathsf{NTCF},\Gen_{\mathcal{F}}(1^\lambda)}{\Pr}\left[ \mathcal{A}(k) \in H_k \right] - \underset{(k,\mathsf{td}) \leftarrow \mathsf{NTCF},\Gen_{\mathcal{F}}(1^\lambda)}{\Pr}\left[ \mathcal{A}(k) \in \overline{H}_k \right] \,
\vline \leq \negl(\lambda).
$$

\end{itemize}
\end{description}
\end{definition}

Brakerski et al.~\cite{JACM:BCMVV21} showed the following theorem.

\begin{theorem}[\cite{JACM:BCMVV21}] Assuming the hardness of the LWE problem, there exists an NTCF family.
\end{theorem}

We also recall the following amplified adaptive hardcore bit property~\cite{CoRR:RadSat19,TCC:KitNisYam21}.

\begin{definition}[Amplified Adaptive Hardcore Property]\label{def:ampl-adaptive-hc}
We say that a NTCF family $\mathcal{F}$ satisfies the amplified adaptive hardcore bit property if, for any QPT $\mathcal{A}$ and $n=\omega(\log \lambda)$, it holds that
\begin{align}
\Pr \left[ 
\begin{array}{cc}
\forall i \in [n] \,: \, x_i = x_{i,b_i},\\
d_i \,\in \, \mathcal{G}_{k,0,x_{i,0}}
\cap \,\mathcal{G}_{k,1,x_{i,1}},\\
c_i = d_i \cdot (J(x_{i,0}) \oplus J(x_{i,1}))
\end{array}
:
\begin{array}{ll}
\forall i \in [n] \,: \, (k_i,\mathsf{td}_i) \leftarrow \mathsf{NTCF}.\Gen_{\mathcal{F}}(1^\lambda)\\
\{(b_i,x_i,y_i, d_i,c_i)\}_{i \in [n]} \leftarrow \mathcal{A}(\{k_i\}_{i \in [n]})\\
x_{i,\beta} \leftarrow \mathsf{Inv}_{\mathcal{F}}(\mathsf{td}_i,\beta,y_i) \, \text{ for } (i,\beta) \in [n] \times \bit
\end{array}
\right] \leq \negl(\lambda).
\end{align} 
\end{definition}

Finally, we make use of the following result.

\begin{lemma}[\cite{CoRR:RadSat19,TCC:KitNisYam21}]
Any NTCF family satisfies the amplified adaptive hardcore property.
\end{lemma}

\section{Two-tier One-shot Signatures}
\label{sec:2-OSS}


In this section, we define two-tier one-shot signatures (2-OSS),
and construct it from the LWE assumption~\cite{Reg05}. Broadly speaking, this cryptographic primitive is a variant of one-shot signatures~\cite{STOC:AGKZ20},
where the verification of a signature for the message ``$0$'' is done publicly, whereas that for the message ``$1$'' is done
only privately.

\subsection{Definition}
The formal definition of 2-OSS is as follows.
\begin{definition}[Two-Tier One-Shot Signatures (2-OSS)]\label{def:two-tier-one-shot-signature}
A two-tier one-shot signature scheme is a set \break $(\Setup,\KeyGen,\Sign,\Ver_0,\Ver_1)$ of algorithms such that
\begin{itemize}
\item 
$\Setup(1^\secp)\to(\pp,\sk):$
It is a QPT algorithm that, on input the security parameter $\secp$, outputs a classical parameter $\pp$ and a classical secret key $\sk$.
    \item 
    $\KeyGen(\pp)\to (\sigk,\vk):$
    It is a QPT algorithm that, on input $\pp$, outputs
    a quantum signing key $\sigk$ and a classical verification key $\vk$.

    \item 
    $\Sign(\sigk,m)\to\sigma:$
    It is a QPT algorithm that, on input $\sigk$ and a message $m\in\bit$, outputs a classical signature $\sigma$.
    \item 
    $\Ver_0(\pp,\vk,\sigma)\to\top/\bot:$
    It is a QPT algorithm that, on input $\pp$, $\vk$, and $\sigma$, outputs $\top/\bot$.
    \item 
    $\Ver_1(\pp,\sk,\vk,\sigma)\to\top/\bot:$
    It is a QPT algorithm that, on input $\pp$, $\sk$, and $\sigma$, outputs $\top/\bot$. 
\end{itemize}
We require the following properties.

\paragraph{Correctness:}
\begin{align}
\Pr\left[
\top\gets\Ver_0(\pp,\vk,\sigma):
\begin{array}{rr}
(\sk,\pp)\gets\Setup(1^\secp)\\
(\sigk,\vk)\gets\KeyGen(\pp)\\
\sigma\gets\Sign(\sigk,0)
\end{array}
\right]\ge1-\negl(\secp)    
\end{align}
and
\begin{align}
\Pr\left[\top\gets\Ver_1(\pp,\sk,\vk,\sigma):
\begin{array}{rr}
(\sk,\pp)\gets\Setup(1^\secp)\\
(\sigk,\vk)\gets\KeyGen(\pp)\\
\sigma\gets\Sign(\sigk,1)
\end{array}
\right]\ge1-\negl(\secp).
\end{align}

\paragraph{Security:}
For any QPT adversary $\cA$,
\begin{align}
\Pr\left[
\top\gets\Ver_0(\pp,\vk,\sigma_0)
\wedge
\top\gets\Ver_1(\pp,\sk,\vk,\sigma_1)
:
\begin{array}{rr}
(\sk,\pp)\gets\Setup(1^\secp)\\
(\vk,\sigma_0,\sigma_1)\gets\cA(\pp)
\end{array}
\right]\le\negl(\secp).
\end{align}
\end{definition}

\subsection{Construction}
We show that 2-OSS can be constructed from the LWE assumption~\cite{Reg05}. Specifically, we make use of NTCF families (see \Cref{sec:ntcf}) which allow us to generate quantum states that have a nice structure in both the computational basis, as well as the Hadamard basis. Our 2-OSS scheme is based on the two-tier quantum lightning scheme in \cite{TCC:KitNisYam21} and leverages this structure to sign messages: to sign the message ``$0$'', we output a measurement outcome in the computational basis, whereas if we wish to sign ``$1$'', we output a measurement outcome in the Hadamard basis.
Crucially, the so-called adaptive hardcore-bit property (see \Cref{def:NTCF}) ensures that it is computationally difficult to produce the two outcomes simultaneously.

\begin{theorem}
Assuming the quantum hardness of the LWE problem, there exists two-tier one-shot signatures.  
\end{theorem}

\begin{proof}
To construct a 2-OSS scheme from LWE, we use NTCF families (see \Cref{def:NTCF}). Our construction is almost identical to the construction of two-tier quantum lightning in \cite{TCC:KitNisYam21}. 

\paragraph{Construction.} 
Let $n = \omega(\log \lambda)$. Consider the following scheme.
\begin{itemize}

    \item $\Setup(1^\secp)$: Generate $(k_i,\mathsf{td}_i) \leftarrow \mathsf{NTCF}.\Gen_{\mathcal{F}}(1^\lambda)$ for $i \in [n]$ and set $(\pp,\sk):=(\{k_i\}_{i \in [n]}, \{ \mathsf{td}_i\}_{i \in [n]})$.
   
    \item 
    $\KeyGen(\pp)$:
    Parse $\pp = \{k_i\}_{i \in [n]}$. For each $i \in [n]$, generate a quantum state
$$
    \ket{\psi'_i} = \frac{1}{\sqrt{|\mathcal{X}|}} \sum_{x \in \mathcal{X}, y\in \mathcal{Y},
    b \in \bit}
\sqrt{(f'_{k_i,b}(x))(y)} \ket{b,x} \ket{y}
    $$
by using $\mathsf{Samp}_{\mathcal{F}}$, measure the last register to obtain $y_i \in \mathcal{Y}$, and let $\ket{\phi'_i}$ be the post-measurement state where the measured register is discarded. Output $(\vk,\sigk) = (\{y_i\}_{i \in [n]}, \{\ket{\phi'_i}\}_{i \in [n]})$.    
    \item 
    $\Sign(\sigk,m):$ on input $m \in \bit$, parse $\sigk$ as $ \{\ket{\phi'_i}\}_{i \in [n]}$ and then proceed as follows:
    \begin{itemize}
        \item if $m=0$: measure $\ket{\phi'_i}$ in the computational basis for $i \in [n]$, which results in measurement outcomes $\{(b_i,x_i)\}_{i \in [n]}$, and output $\sigma = \{(b_i,x_i)\}_{i \in [n]}$ as the signature.

         \item if $m=1$: measure $\ket{\phi'_i}$ in the Hadamard basis for $i \in [n]$, which results in measurement outcomes $\{(c_i,d_i)\}_{i \in [n]}$, and output $\sigma = \{(c_i,d_i)\}_{i \in [n]}$ as the signature.
    \end{itemize}
    
    \item 
    $\Ver_0(\pk,\vk,\sigma):$
    Parse $\pp = \{k_i\}_{i \in [n]}$, $\vk = \{y_i\}_{i \in [n]}$ and  $\sigma = \{(b_i,x_i)\}_{i \in [n]}$. Use the procedure $\mathsf{Chk}_{\mathcal{F}}$ in \Cref{def:NTCF} to determine if $y_i \in \mathsf{Supp}(f'_{k,b}(x_i))$ for all $i \in [n]$. Output $\top$ if it is the case, and output $\bot$ otherwise.
    \item 
    $\Ver_1(\pp,\sk,\vk,\sigma):$
    Parse $\pp = \{k_i\}_{i \in [n]}$, $\sk = \{ \mathsf{td}_i\}_{i \in [n]}$, $\vk = \{y_i\}_{i \in [n]}$  and $\sigma = \{(c_i,d_i)\}_{i \in [n]}$. Then, compute $x_{i,\beta} \leftarrow \mathsf{Inv}_{\mathcal{F}}(\mathsf{td}_i, \beta, y_i)$ for every $i\in [n]$ and $\beta \in \bit$. Output $\top$, if $d_i \in \mathcal{G}_{k,0,x_{i,0}} \cap \mathcal{G}_{k,1,x_{i,1}}$ and $c_i = d_i \cdot (J(x_{i,0}) \oplus J(x_{i,1}))$ for all $i \in [n]$, and output $\bot$ otherwise.
\end{itemize}
The correctness of our 2-OSS scheme follows from the properties of the NTCF in \Cref{def:NTCF}. Security follows from the amplified adaptive hardcore bit property in \Cref{def:ampl-adaptive-hc}.
\end{proof}

\section{Digital Signatures with Revocable Signing Keys}
\label{sec:skcert}
In this section, we define digital signatures with revocable signing keys ($\mathsf{DSR\mbox{-}Key}$) and 
give its construction from 2-OSS.

\subsection{Definition}

Let us now present a formal definition of $\mathsf{DSR\mbox{-}Key}$. Note that we consider the \emph{stateful} setting which requires that the signer keep a \emph{state} of all previously signed messages and keys.

\begin{definition}[(Stateful) Digital Signatures with Revocable Signing Keys ($\mathsf{DSR\mbox{-}Key}$)]\label{def:dswrsk}
A (stateful) digital signature scheme with revocable signing keys is the following set $ (\setup,\KeyGen,\Sign,\Ver,\Del,\Cert)$    
of algorithms:
\begin{itemize}
\item 
$\setup(1^\secp)\to(\ck,\pp):$
It is a QPT algorithm that, on input the security parameter $\secp$, outputs a classical key $\ck$ and a classical parameter $\pp$.
    \item 
    $\KeyGen(\pp)\to(\sigk_0,\vk):$
    It is a QPT algorithm that, on input $\pp$, outputs 
    a quantum signing key $\sigk_0$ and a classical verification key $\vk$.

    \item 
    $\Sign(\pp,\sigk_{i}, m)\to(\sigk_{i+1},\sigma):$
    It is a QPT algorithm that, on input $\pp$, a message $m$ and a signing key $\sigk_{i}$, outputs a subsequent signing key $\sigk_{i+1}$ and a classical signature $\sigma$. 
    \item 
    $\Ver(\pp,\vk,m,\sigma)\to\top/\bot:$
    It is a QPT algorithm that, on input $\pp$, $\vk$, $m$, and $\sigma$, outputs $\top/\bot$.
    \item 
    $\Del(\sigk_i)\to\cert:$
    It is a QPT algorithm that, on input $\sigk_i$, outputs a classical certificate $\cert$. 
    \item 
    $\Cert(\pp,\vk,\ck,\cert,S)\to\top/\bot:$
    It is a QPT algorithm that, on input $\pp$, $\vk$, $\ck$, $\cert$, and a set $S$ consisting of messages, outputs $\top/\bot$. 
\end{itemize}
We require the following properties.

\paragraph{Many-time correctness:}
For any polynomial $p=p(\lambda)$, and any messages $(m_1,m_2,...,m_{p})$,
\begin{align}
\Pr\left[\bigwedge_{i\in[p]}\top\gets\Ver(\pp,\vk,m_i,\sigma_i):
\begin{array}{rr}
(\pp,\ck)\gets\setup(1^\secp)\\
     (\sigk_0,\vk)\gets\KeyGen(\pp)  \\
     (\sigk_1,\sigma_1)\gets\Sign(\pp,\sigk_0,m_1)\\
     (\sigk_2,\sigma_2)\gets\Sign(\pp,\sigk_1,m_2)\\
     ...\\
     (\sigk_{p},\sigma_{p})\gets\Sign(\pp,\sigk_{p-1},m_{p})
\end{array}
\right]\ge1-\negl(\secp).    
\end{align}
We say that the scheme satisfies one-time correctness if the above property is satisfied for $p=1$.

\paragraph{EUF-CMA security:}
For any QPT adversary $\cA$,
\begin{align}
\Pr\left[\top\gets\Ver(\pp,\vk,m^*,\sigma^*):
\begin{array}{rr}
(\pp,\ck)\gets\setup(1^\secp)\\
(\sigk_0,\vk)\gets\KeyGen(\pp)\\
(m^*,\sigma^*)\gets\cA^{\cO_{\Sign}}(\vk)
\end{array}
\right]    
\le\negl(\secp),
\label{revocable_sigk_EUF-CMA}
\end{align}
where $\cO_{\Sign}$ is a stateful signing oracle defined below and 
$\cA$ is not allowed to query the oracle on $m^*$: 
\begin{description}
\item[$\cO_{\Sign}$:]
Its initial state is set to be $(\pp,\sigk_0)$.   
When a message $m$ is queried, it proceeds as follows:
\begin{itemize}
\item Parse its state as $(\pp,\sigk_i)$.
\item Run $(\sigk_{i+1},\sigma)\gets \Sign(\pp,\sigk_i,m)$.
\item Return $\sigma$ to $\cA$ and update its state to $(\pp,\sigk_{i+1})$.
\end{itemize}
\end{description} 
We say that the scheme satisfies one-time EUF-CMA security if \Cref{revocable_sigk_EUF-CMA} holds for any $\cA$ that submits at most one query to the oracle.

\paragraph{Many-time deletion correctness:}
For any polynomial $p=p(\secp)$, and any messages $(m_1,m_2,...,m_{p})$,
\begin{align}
\Pr\left[\top\gets\Cert(\pp,\vk,\ck,\cert,\{m_1,m_2,...,m_p\}):
\begin{array}{rr}
(\pp,\ck)\gets\setup(1^\secp)\\
     (\sigk_0,\vk)\gets\KeyGen(\pp)  \\
     (\sigk_1,\sigma_1)\gets\Sign(\pp,\sigk_0,m_1)\\
     (\sigk_2,\sigma_2)\gets\Sign(\pp,\sigk_1,m_2)\\
     ...\\
     (\sigk_{p},\sigma_{p})\gets\Sign(\pp,\sigk_{p-1},m_{p})\\
     \cert\gets\Del(\sigk_{p})
\end{array}
\right]\ge1-\negl(\secp). 
\end{align} 
We remark that  we require the above to also hold for the case of $p=0$, in which case the fifth component of the input of $\Cert$ is the empty set $\emptyset$. 
We say that the scheme satisfies one-time deletion correctness if the above property is satisfied for $p\le 1$.

\paragraph{Many-time deletion security:}
For any QPT adversary $\cA$, 
\begin{align}
   \Pr\left[
   \begin{array}{ll}
   &~~~\top\gets\Cert(\pp,\vk,\ck,\cert,S)\\
   &\wedge~ m^*\notin S\\
   &\wedge~ \top\gets\Ver(\pp,\vk,m^*,\sigma^*)
   \end{array}
   :
   \begin{array}{rr}
   (\pp,\ck)\gets\setup(1^\secp)\\
   (\vk,\cert,S,m^*,\sigma^*)\gets\cA(\pp)\\
   \end{array}
   \right] 
   \le\negl(\secp).
\end{align}
We say that the scheme satisfies one-time deletion security if the above property is satisfied if we additionally require $|S|\le 1$.
\end{definition}

\begin{remark}\label{rem:definition_revocable_signing_key}
Following the definition of single signer security in \cite{STOC:AGKZ20} or 
copy-protection security in \cite{TCC:LLQZ22}, it is also reasonable to define deletion security as follows:

For any pair $(\cA_1,\cA_2)$ of QPT adversaries and any distribution $\cD$ with super-logarithmic min-entropy  over the message space,  
\begin{align}
   \Pr\left[\top\gets\Cert(\pp,\vk,\ck,\cert,S) \wedge \top\gets\Ver(\pp,\vk,m,\sigma):
   \begin{array}{rr}
   (\pp,\ck)\gets\setup(1^\secp)\\
   (\vk,\cert,S,\st)\gets\cA_1(\pp)\\
   m\gets\cD\\
   \sigma\gets\cA_2(m,\st)
   \end{array}
   \right] 
   \le\negl(\secp).
\end{align} 
It is easy to see that our definition implies the above, but the converse is unlikely. This is  why we define deletion security as in \Cref{def:dswrsk}.
\end{remark}

\subsection{One-Time Construction for Single-Bit Messages}
\label{sec:onetime}

Here we show that we can construct one-time $\mathsf{DSR\mbox{-}Key}$ for single-bit messages 
from 2-OSS in a black-box way.
\begin{theorem}
If two-tier one-shot signatures exist, then digital signatures with revocable signing keys with the message space $\bit$  that satisfy one-time variants of correctness, EUF-CMA security, deletion correctness, and deletion security in \Cref{def:dswrsk}
exist. 
\end{theorem}

\begin{proof}
Let $(\mathsf{OS}.\setup,\mathsf{OS}.\KeyGen,\mathsf{OS}.\Sign,\mathsf{OS}.\Ver_0,\mathsf{OS}.\Ver_1)$ be a two-tier one-shot signature scheme.
From it, we construct a one-time digital signature scheme $\Sigma\coloneqq(\setup,\KeyGen,\Sign,\Ver,\Del,\Cert)$ with revocable signing keys  as follows.
\begin{itemize}
\item 
$\setup(1^\secp)\to(\ck,\pp):$
Run $(\pp',\sk)\gets\mathsf{OS}.\setup(1^\secp)$.
Output $\ck\coloneqq\sk$ and $\pp\coloneqq \pp'$.
    \item 
    $\KeyGen(\pp)\to(\sigk,\vk):$
    Run $(\sigk_0,\vk_0)\gets\mathsf{OS}.\KeyGen(\pp)$.
    Run $(\sigk_1,\vk_1)\gets\mathsf{OS}.\KeyGen(\pp)$.
    Output $\sigk\coloneqq (\sigk_0,\sigk_1)$
    and $\vk\coloneqq (\vk_0,\vk_1)$.
    \item 
    $\Sign(\pp,\sigk,m)\to(\sigk',\sigma):$
    Parse $\sigk=(\sigk_0,\sigk_1)$. Run $\sigma \gets\mathsf{OS}.\Sign(\pp,\sigk_{m},0)$ and let $\sigk'\coloneqq(1-m,\sigk_{1-m})$. Output $(\sigk',\sigma)$.
    \item 
    $\Ver(\pp,\vk,m,\sigma)\to\top/\bot:$
    Parse $\vk=(\vk_0,\vk_1)$.
    Run $\mathsf{OS}.\Ver_0(\pp,\vk_m,\sigma)$,
    and output its output.
    \item 
    $\Del(\sigk)\to\cert:$ 
    \begin{itemize}
    \item If $\sigk$ is of the form $(\sigk_0,\sigk_1)$ (i.e., no signature has been generated),  generate $\sigma'_0 \gets\mathsf{OS}.\Sign(\sigk_0,1)$ 
    and 
    $\sigma'_1 \gets\mathsf{OS}.\Sign(\sigk_1,1)$ 
    and output $\cert = (0,1,\sigma'_0,\sigma'_1)$.
    \item Otherwise, parse $\sigk = (1-m, \sigk_{1-m})$, generate $\sigma' \gets\mathsf{OS}.\Sign(\sigk_{1-m},1)$ and output $\cert = (1-m,\sigma')$.
    \end{itemize}

    \item 
    $\Cert(\pp,\vk,\ck,\cert,S)\to\top/\bot:$
    Parse $\vk=(\vk_0,\vk_1)$ and $\ck=\sk$.
    \begin{itemize}
    \item If $S=\emptyset$, parse $\cert=(0,1,\sigma'_0,\sigma'_1)$. (Output $\bot$ if $\cert$ is not of this form.)  
    Output $\top$ if $\mathsf{OS}.\Ver_1(\pp,\sk,\vk_b,\sigma'_b)=\top$ for both $b=0,1$. 
    \item If $S=\{m\}$ for some $m\in \bit$, 
    parse $\cert=(1-m,\sigma')$. (Output $\bot$ if $\cert$ is not of this form.)
    Output $\top$ if $\mathsf{OS}.\Ver_1(\pp,\sk,\vk_{1-m},\sigma')=\top$
    \item If $S=\{0,1\}$, output $\bot$. 
    \end{itemize}
\end{itemize}

Let us now verify that all of the one-time properties are satisfied. 

\paragraph{One-time correctness.}

Let $m \in \bit$ be an arbitrary message. Then, the correctness of $\mathsf{OS}.\Ver_0$ immediately implies that the following one-time correctness holds, i.e.,
\begin{align}
\Pr\left[\top \leftarrow\Ver(\pp,\vk,m,\sigma):
\begin{array}{rr}
(\pp,\ck)\gets\setup(1^\secp)\\
     (\sigk,\vk)\gets\KeyGen(\pp)  \\
     (\sigk',\sigma)\gets\Sign(\pp,\sigk,m)
\end{array}
\right]\ge1-\negl(\secp). 
\end{align}

\paragraph{One-time EUF-CMA security.} Suppose that there exists a QPT adversary $\cA$ which violates one-time EUF-CMA security. In other words, there is a polynomial $q(\lambda)$ such that for infinitely many $\lambda \in \N$,
\begin{align}\label{eq:one-time-EUF-CMA-TTOS}
\Pr\left[\top\gets\Ver(\pp,\vk,m^*,\sigma^*):
\begin{array}{rr}
(\pp,\ck)\gets\setup(1^\secp)\\
(\sigk,\vk)\gets\KeyGen(\pp)\\
(m^*,\sigma^*)\gets\cA^{\cO_{\Sign}}(\vk)
\end{array}
\right]    
\ge \frac{1}{q(\secp)}.
\end{align}
where $\cA$ only makes at most single query  to the stateful $\cO_{\Sign}$ in \Cref{def:dswrsk} which does not equal $m^*$. Without loss of generality, we assume that $\cA$ queries on $1-m^*$ (otherwise, if $\cA$ does not submit any query, we can simply consider $\cA'$ which first runs $\cA$ and then subsequently queries on $1-m^*$). Then, we can construct the following QPT algorithm $\cB$ which breaks the security of the two-tier one-shot signature scheme:
\begin{enumerate}
    \item On input $\pp$, $\cB$ generates $(\sigk,\vk)\gets \KeyGen(\pp)$ and parses the signing key as $\sigk = (\sigk_0,\sigk_1)$
    and the verification key as $\vk= (\vk_0,\vk_1)$.

    \item $\cB$ runs the adversary $\cA$ on input $\vk$; when $\cA$ submits a query, say of the form $1-m^*$, $\cB$ simulates the oracle $\cO_{\Sign}$ by running $(\sigk',\sigma') \gets \Sign(\pp,\sigk,1-m^*)$ and sends back $\sigma'$ to $\cA$. Let $(m^*,\sigma^*)$ denote the final output returned by $\cA^{\cO_{\Sign}}(\vk)$ and assign $\sigma_0 = \sigma^*$.

    \item $\cB$ parses $\sigk' = (m^*,\sigk_{m^*})$ and computes $\sigma_1 \gets \mathsf{OT}.\Sign(\pp,\sigk_{m^*},1)$.
    \item $\cB$ outputs the triplet $(\vk_{m^*},\sigma_0,\sigma_1)$. 
\end{enumerate}
We now argue that $\cB$ breaks the security of the two-tier one-shot signature scheme. Using \Cref{eq:one-time-EUF-CMA-TTOS}, it follows that $\cB$ obtains two valid signatures $\sigma_0,\sigma_1$ (of which one corresponds to $0$ and the other corresponds to $1$) with probability at least
\begin{align}
\Pr\left[
\top\gets\mathsf{OS}.\Ver_0(\pp,\vk,\sigma_0)
\wedge
\top\gets \mathsf{OS}.\Ver_1(\pp,\sk,\vk,\sigma_1)
:
\begin{array}{rr}
(\sk,\pp)\gets \mathsf{OS}.\Setup(1^\secp)\\
(\vk,\sigma_0,\sigma_1)\gets\cB(\pp)
\end{array}
\right]\ge \frac{1}{\poly(\secp)}.
\end{align}
This violates the security of the two-tier one-shot signature according to \Cref{def:two-tier-one-shot-signature}.

\paragraph{One-time deletion correctness.}
Here, we consider two cases. First, when no signature is produced, $\sigk$ is of the form $(\sigk_0,\sigk_1)$ and $\cert = (0,1,\sigma'_0,\sigma'_1)$. Thus, both 
$\sigma'_0 \gets\mathsf{OS}.\Sign(\sigk_0,1)$ 
    and 
    $\sigma'_1 \gets\mathsf{OS}.\Sign(\sigk_1,1)$ pass verification via $\mathsf{OS}.\Ver_1(\pp,\sk,\vk_0,\cdot)$ and $\mathsf{OS}.\Ver_1(\pp,\sk,\vk_1,\cdot)$ with overwhelming probability.
Therefore, we get the following from the union bound and the correctness of $\mathsf{OS}.\Ver_1$ that
\begin{align}
\Pr\left[\top\gets\Cert(\pp,\vk,\ck,\cert,\emptyset):
\begin{array}{rr}
(\pp,\ck)\gets\setup(1^\secp)\\
     (\sigk,\vk)\gets\KeyGen(\pp) \\
     \cert\gets\Del(\sigk)
\end{array}
\right]\ge1-\negl(\secp). 
\end{align} 
Next, we analyze the case when precisely one signature is generated on some message $m \in \bit$. In this case, $\sigk'$ is of the form
$(1-m, \sigk_{1-m})$ and $\cert=(1-m,\sigma')$, and hence $\sigma' \gets\mathsf{OS}.\Sign(\sigk_{1-m},1)$ passes verification via $\mathsf{OS}.\Ver_1(\pp,\sk,\vk_{1-m},\cdot)$ with overwhelming probability. This implies that
\begin{align}
\Pr\left[\top\gets\Cert(\pp,\vk,\ck,\cert,\{m\}):
\begin{array}{rr}
(\pp,\ck)\gets\setup(1^\secp)\\
     (\sigk,\vk)\gets\KeyGen(\pp)  \\
     (\sigk',\sigma)\gets\Sign(\pp,\sigk,m)\\
     \cert\gets\Del(\sigk')
\end{array}
\right]\ge1-\negl(\secp). 
\end{align} 

\paragraph{One-time deletion security:}
Suppose that there exists a QPT adversary $\cA$ which violates one-time deletion security. In other words, there exists a polynomial $q(\lambda)$ such that for infinitely many $\lambda \in \N$,
it holds that
\begin{align}
   \Pr\left[
   \begin{array}{ll}
   &~~~\top\gets\Cert(\pp,\vk,\ck,\cert,S)\\
   &\wedge~ m^*\notin S\\
   &\wedge~ \top\gets\Ver(\pp,\vk,m^*,\sigma^*)\\
   &\wedge~|S|\le 1
   \end{array}
   :
   \begin{array}{rr}
   (\pp,\ck)\gets\setup(1^\secp)\\
   (\vk,\cert,S,m^*,\sigma^*)\gets\cA(\pp)\\
   \end{array}
   \right] 
\ge \frac{1}{q(\secp)}.
\end{align}
Then, we can construct the following QPT algorithm $\cB$ which breaks the security of the underlying two-tier one-shot signature scheme according to \Cref{def:two-tier-one-shot-signature}.
\begin{enumerate}
    \item On input $\pp$, $\cB$ runs $\cA(\pp)$ and receives a tuple $(\vk,\cert,S,m^*,\sigma^*)$ for some set $S$ of size $|S| \leq 1$ and some verification key $\vk = (\vk_0,\vk_1)$.
    \item Depending on the size of $S$, $\cB$ then does the following:
    \begin{itemize}
        \item If $|S|=0$, i.e., $S=\emptyset$, $\cB$ parses $\cert = (0,1,\sigma_0',\sigma_1')$ and outputs $(\vk_{m^*},\sigma^*,\sigma_{m^*}')$.

        \item If $|S|=1$, i.e., $S=\{m\}$ for $m \neq m^*$, $\cB$ parses $\cert = (1-m,\sigma')$ and outputs $(\vk_{m^*},\sigma^*,\sigma')$.
    \end{itemize}
\end{enumerate}
Note that in either case, whether $|S|=0$ or $|S|=1$, it follows that $\cB$ succeeds with probability at least
\begin{align*}
\Pr\left[
\top\gets\mathsf{OS}.\Ver_0(\pp,\vk,\sigma_0)
\wedge
\top\gets \mathsf{OS}.\Ver_1(\pp,\sk,\vk,\sigma_1)
:
\begin{array}{rr}
(\sk,\pp)\gets \mathsf{OS}.\Setup(1^\secp)\\
(\vk,\sigma_0,\sigma_1)\gets\cB(\pp)
\end{array}
\right]\ge \frac{1}{\poly(\secp)},
\end{align*}
which violates the security of the two-tier one-shot signature according to \Cref{def:two-tier-one-shot-signature}.
\end{proof}

\subsection{From Single-Bit to Multi-Bit Messages}
Here, we show that we can expand the message space to $\bit^*$ using collision-resistant hashes.
\begin{theorem}
If  
collision-resistant hash functions and 
digital signatures with revocable signing keys with the message space $\bit$  that satisfy one-time variants of correctness, EUF-CMA security, deletion correctness, and deletion security in \Cref{def:dswrsk}
exist, then a similar scheme with the message space $\bit^*$ exists. 
\end{theorem}

\begin{proof}
Let $\Sigma=(\setup,\KeyGen,\Sign,\Ver,\Del,\Cert)$ be a scheme with the message space $\bit$ which satisfies the one-time variants of correctness, EUF-CMA security, deletion correctness, and deletion security.
Let $\mathcal{H}$ be a family of collision resistant hash functions from $\bit^*$ to $\bit^\ell$. 
Then we construct $\Sigma'=(\setup',\KeyGen',\Sign',\Ver',\Del',\Cert')$
with the message space $\bit^*$ 
as follows:
\begin{itemize}
    \item $\setup'(1^\secp) \rightarrow (\ck,\pp)$:
    Choose $H\gets \mathcal{H}$. 
    Run $(\Sigma.\ck,\Sigma.\pp)\gets \setup(1^\secp)$ and output $\ck:=\Sigma.\ck$ and $\pp:=(\Sigma.\pp,H)$. 
    \item $\KeyGen'(\pp) \rightarrow (\sigk,\vk)$: 
    Parse $\pp=(\Sigma.\pp,H)$. 
    For $i\in[\ell]$, 
    run 
    $(\sigk_i,\vk_i) \leftarrow \KeyGen(\Sigma.\pp)$ and output $\sigk:=\{\sigk_i\}_{i\in [\ell]}$ as the quantum signing key and $\vk:=\{\vk_i\}_{i\in [\ell]}$ as the classical verification key.

 \item 
    $\Sign'(\pp,\sigk,m)\to(\sigk',\sigma):$
  Parse
  $\pp=(\Sigma.\pp,H)$
  and 
  $\sigk=\{\sigk_i\}_{i\in [\ell]}$
  Compute $m'=H(m)$ and let $m'_i$ be the $i$th bit of $m'$ for $i\in[\ell]$. 
  For $i\in[\ell]$, 
    run 
    $(\sigk'_i,\sigma_i) \leftarrow \Sign(\Sigma.\pp,\sigk_i,m'_i)$ and output $\sigk':=\{\sigk'_i\}_{i\in [\ell]}$ and $\sigma:=\{\sigma_i\}_{i\in [\ell]}$.  
    \item $\Ver'(\pp,\vk,m,\sigma)\to\top/\bot:$ Parse 
    $\pp=(\Sigma.\pp,H)$, $\vk:=\{\vk_i\}_{i\in [\ell]}$, and $\sigma=\{\sigma_i\}_{i\in [\ell]}$. 
    Compute $m'=H(m)$ and let $m'_i$ be the $i$th bit of $m'$ for $i\in[\ell]$.  
    Output $\top$ of $\Ver(\Sigma.\pp,\vk_i,m'_i,\sigma_i)=\top$ for every $i\in [\ell]$. 
    \item 
    $\Del'(\sigk)\to\cert:$ 
Parse $\sigk=\{\sigk_i\}_{i\in [\ell]}$. 
For $i\in[\ell]$, run $\cert_i\gets \Del(\sigk_i)$ and output $\cert:=\{\cert_i\}_{i\in [\ell]}$.  
    \item 
    $\Cert'(\pp,\vk,\ck,\cert,S)\to\top/\bot:$
    Parse
  $\pp=(\Sigma.\pp,H)$,  
  $\vk=\{\vk_i\}_{i\in [\ell]}$, 
   and 
    $\cert=\{\cert_i\}_{i\in [\ell]}$.
    \begin{itemize}
    \item If $S=\emptyset$, output $\top$ if $\Cert(\Sigma.\pp,\vk_i,\ck,\cert_i,\emptyset)=\top$ for every $i\in [\ell]$. 
    \item If $S=\{m\}$ for some $m\in \bit^*$, compute $m'=H(m)$ and let $m'_i$ be the $i$th bit of $m'$ for $i\in[\ell]$. Output $\top$ if  
 $\Cert(\Sigma.\pp,\vk_i,\ck,\cert_i,\{m'_i\})=\top$ for every $i\in [\ell]$. 
    \item Otherwise (i.e., if $|S|\geq 2$), output $\bot$. 
    \end{itemize}
\end{itemize}
The proof of correctness is immediate and the proof of one-time EUF-CMA security follows from standard techniques which allow conventional signature schemes to handle messages of arbitrarily length, see~\cite{books/crc/KatzLindell2007} for example. 
Therefore, it suffices to show that the scheme $\Sigma'$ satisfies the one-time variants of deletion correctness and deletion security. Let us now verify each property separately.

\paragraph{One-time deletion correctness.} Because $\Sigma$ satisfies one-time correctness, this means that there exists a negligible function $\nu$ such that for any $p \le 1$ and any message $m_p$ (possibly $m_p = \emptyset$, if $p=0$, in which case no message is signed in the expression below) and for all $\lambda \in \N$,
\begin{align}\label{eq:ot-deletion-correctness}
\Pr\left[\top\gets\Sigma.\Cert(\pp,\vk,\ck,\cert,\{m_p\}):
\begin{array}{rr}
(\pp,\ck)\gets\Sigma.\setup(1^\secp)\\
     (\sigk_0,\vk)\gets\Sigma.\KeyGen(\pp)  \\
     (\sigk_p,\sigma_p)\gets\Sigma.\Sign(\pp,\sigk_0,m_p)\\
     \cert\gets \Sigma.\Del(\sigk_{p})
\end{array}
\right]\ge1-\nu(\secp). 
\end{align} 
Hence, by the union bound and the fact that $\ell\cdot \nu(\lambda) = \negl(\lambda)$ with $\ell = \ell(\lambda)=\poly(\lambda)$, it follows that $\Sigma'$ also satisfies one-time deletion correctness.

\paragraph{One-time deletion security.} Suppose that $\Sigma'$ does not satisfy one-time deletion security. In other words, there exists a QPT adversary $\cA$ and a polynomial $p(\lambda)$ such that for infinitely many $\lambda \in \N$, it holds that: 
\begin{align}\label{adv-against-OT-del-security}
   \Pr\left[
   \begin{array}{ll}
   &~~~\top\gets \Sigma'.\Cert(\pp,\vk,\ck,\cert,S)\\
   &\wedge~ m^*\notin S\\
   &\wedge~ \top\gets \Sigma'.\Ver(\pp,\vk,m^*,\sigma^*)\\
   &\wedge |S|\leq 1
   \end{array}
   :
   \begin{array}{rr}
   (\pp,\ck)\gets \Sigma'.\setup(1^\secp)\\
   (\vk,\cert,S,m^*,\sigma^*)\gets\cA(\pp)\\
   \end{array}
   \right] 
   \ge \frac{1}{p(\secp)}.
\end{align}
In the following, we parse $\pp \gets \Sigma'.\setup(1^\secp)$ as $\pp=(\Sigma.\pp,H)$ and distinguish between two cases:

\begin{itemize}
    \item[\textbf{Case 1:}] A collision occurs; namely, when $\cA(\pp)$ outputs $(\vk,\cert,S,m^*,\sigma^*)$, there exists a message $m \in S$ such that $H(m) = H(m^*)$.

    \item[\textbf{Case 2:}] No collision occurs; in other words, when $\cA(\pp)$ outputs $(\vk,\cert,S,m^*,\sigma^*)$, there does not exist a message $m \in S$ such that $H(m) = H(m^*)$.
\end{itemize}
Note that \Cref{adv-against-OT-del-security} implies that either \textbf{Case 1} or \textbf{Case 2} occurs with probability at least $1/p(\lambda)$.

Suppose that \textbf{Case 1} occurs, i.e., it holds that
\begin{align}\label{adv-against-OT-del-security-Case-I}
   \Pr\left[
   \begin{array}{ll}
   &~~~\top\gets \Sigma'.\Cert(\pp,\vk,\ck,\cert,S)\\
   &\wedge~ m^*\notin S\\
   &\wedge~ \top\gets \Sigma'.\Ver(\pp,\vk,m^*,\sigma^*)\\
   &\wedge |S|\leq 1\\
   &\wedge (\exists m \in S \text{ s.t. } H(m)=H(m^*))
   \end{array}
   :
   \begin{array}{rr}
   (\pp,\ck)\gets \Sigma'.\setup(1^\secp)\\
   (\Sigma.\pp,H) = \pp \\(\vk,\cert,S,m^*,\sigma^*)\gets\cA(\pp)\\
   \end{array}
   \right] 
   \ge \frac{1}{p(\secp)}.
\end{align}
We can then construct the following QPT algorithm which breaks the collision-resistance of $\mathcal{H}$:
\begin{enumerate}
    \item On input $(H,1^\lambda)$ run $(\Sigma.\ck,\Sigma.\pp)\gets \setup(1^\secp)$ and output $\ck:=\Sigma.\ck$ and $\pp:=(\Sigma.\pp,H)$. 

    \item Run $\cA(\pp)$ who outputs $(\vk,\cert,S,m^*,\sigma^*)$.

    \item If there exists an $m \in S$ such that $H(m)=H(m^*)$, output $(m,m^*)$. Otherwise, output $\bot$.
\end{enumerate}
Hence, by our assumption in \Cref{adv-against-OT-del-security-Case-I}, we obtain a collision with probability at least $1/p(\secp)$, thereby breaking the collision-resistance of $\mathcal{H}$.

Suppose that \textbf{Case 2} occurs, i.e., it holds that
\begin{align}\label{adv-against-OT-del-security-Case-II}
   \Pr\left[
   \begin{array}{ll}
   &~~~\top\gets \Sigma'.\Cert(\pp,\vk,\ck,\cert,S)\\
   &\wedge~ m^*\notin S\\
   &\wedge~ \top\gets \Sigma'.\Ver(\pp,\vk,m^*,\sigma^*)\\
   &\wedge |S|\leq 1\\
   &\wedge (\nexists m \in S \text{ s.t. } H(m)=H(m^*))
   \end{array}
   :
   \begin{array}{rr}
   (\pp,\ck)\gets \Sigma'.\setup(1^\secp)\\
   (\Sigma.\pp,H) = \pp \\(\vk,\cert,S,m^*,\sigma^*)\gets\cA(\pp)\\
   \end{array}
   \right] 
   \ge \frac{1}{p(\secp)}.
\end{align}
We can then construct the following QPT algorithm which breaks the one-time deletion security of $\Sigma$.
\begin{enumerate}
    \item On input $\Sigma.\pp$, where $(\Sigma.\pp,\Sigma.\ck)\gets \Sigma.\setup(1^\secp)$, sample $H\gets \mathcal{H}$. 
    and let $\pp:=(\Sigma.\pp,H)$. 

    \item Run $(\vk,\cert,S,m^*,\sigma^*)\gets\cA(\pp)$.

    \item Depending on the size of $S$, do the following:
    \begin{itemize}
        \item If $|S|=0$, i.e., $S=\emptyset$, parse $\vk=\{\vk_i\}_{i\in [\ell]}$, $\cert=\{\cert_i\}_{i\in [\ell]}$, and $\sigma^*=\{\sigma_i^*\}_{i\in [\ell]}$. Output $(\vk_j,\cert_j,\emptyset,\bar{m}^*_j,\sigma_j^*)$, where $\bar{m}^*=H(m^*)$ and $j \in [\ell]$ is a random index.

        \item If $|S|=1$, i.e., $S=\{m\}$ for $m \neq m^*$, parse $\vk=\{\vk_i\}_{i\in [\ell]}$, $\cert=\{\cert_i\}_{i\in [\ell]}$, and $\sigma^*=\{\sigma_i^*\}_{i\in [\ell]}$. Output $(\vk_j,\cert_j,\{\bar{m}_j\},\bar{m}^*_j,\sigma_j^*)$, where $\bar{m}=H(m)$ and $\bar{m}^*=H(m^*)$ and where $j$ is an index such that the $j$-th bits of $\bar{m}$ and $\bar{m}^*$ are distinct.
    \end{itemize}
\end{enumerate}
Note that in either case, whether $|S|=0$ or $|S|=1$, our reduction breaks the one-time deletion security of $\Sigma$ with inverse-polynomial success probability; in particular, if $|S|=1$ occurs, $\bar{m}$ and $\bar{m}^*$ will be distinct (and hence differ in at least one bit) since we are in \textbf{Case 2}. This proves the claim.
\end{proof}

\subsection{From One-Time Schemes to Many-Time Schemes}

In this section, we show how to extend any one-time scheme into a proper many-time scheme as in \Cref{def:dswrsk}. The transformation is inspired by the chain-based approach for constructing many-time digital signatures, see~\cite{books/crc/KatzLindell2007} for example.\footnote{We could also use the tree-based construction~\cite{C:Merkle87}, which has a shorter (logarithmic) signature length.  We describe the chain-based construction here for ease of presentation. 
}

Let $\mathsf{OT}=(\mathsf{OT}.\setup,\mathsf{OT}.\KeyGen,\mathsf{OT}.\Sign,\mathsf{OT}.\Ver,\mathsf{OT}.\Del,\mathsf{OT}.\Cert)$ be a scheme which satisfies the one-time variants of correctness, EUF-CMA security, deletion correctness, and deletion security according to in \Cref{def:dswrsk}, and has the message space $\bit^*$.  
Then, we construct $\mathsf{MT}=(\mathsf{MT}.\setup,\mathsf{MT}.\KeyGen,\mathsf{MT}.\Sign,\mathsf{MT}.\Ver,\mathsf{MT}.\Del,\mathsf{MT}.\Cert)$
with the message space $\bit^n$ 
as follows:
\begin{itemize}
    \item $\mathsf{MT}.\setup(1^\secp) \rightarrow (\ck,\pp)$: This is the same as $\mathsf{OT}.\setup$.
    \item $\mathsf{MT}.\KeyGen(\pp) \rightarrow (\sigk,\vk)$: run $(\sfot.\sigk_0,\sfot.\vk_0) \leftarrow \mathsf{OT}.\KeyGen(\pp)$ and output $\sigk:=\sfot.\sigk_0$ as the quantum signing key and $\vk:=\sfot.\vk_0$ as the classical verification key.

 \item 
    $\mathsf{MT}.\Sign(\pp,\sigk_{i},m)\to(\sigk_{i+1},\sigma):$
    on input the public parameter $\pp$, a quantum signing key $\sigk_{i}$, and a message $m \in \bit^n$ proceed as follows:
    \begin{enumerate}
    \item Parse $\sigk_{i}$ as $(
    \sfot.\sigk_{i}, \{\sfot.\sigk'_{j}\}_{j\in \{0,1,...,i-1\}}, 
     \{m_j,\sfot.\vk_{j},\sfot.\sigma_j\}_{j\in [i]})$
    \item Generate  $(\sfot.\sigk_{i+1},\sfot.\vk_{i+1}) \leftarrow \mathsf{OT}.\KeyGen(\pp)$. 
    \item Run
    $$(\sfot.\sigk'_{i},\sfot.\sigma_{i+1}) \leftarrow \mathsf{OT}.\Sign(\pp,\sfot.\sigk_{i}, m\concat \sfot.\vk_{i+1}).$$
    \item 
    Set $m_{i+1}:=m$ and 
    output a subsequent signing key 
    $$\sigk_{i+1}:=(
    \sfot.\sigk_{i+1},  \{\sfot.\sigk'_{j}\}_{j\in \{0,1,...,i\}},
    \{m_j,\sfot.\vk_{j},\sfot.\sigma_j\}_{j\in [i+1]})$$
    and a signature 
    $$\sigma:= \{m_j,\sfot.\vk_{j},\sfot.\sigma_j\}_{j\in [i+1]}.$$
    \end{enumerate}

    \item $\mathsf{MT}.\Ver(\pp,\vk,m,\sigma)\to\top/\bot:$ on input $\pp$, a key $\vk$, a message $m$, and signature $\sigma$, proceed as follows.
    \begin{enumerate}
        \item Parse $\sigma$ as $\{m_j,\sfot.\vk_{j},\sfot.\sigma_j\}_{j\in [i]}$ and let $\sfot.\vk_0=\vk$. 
        \item Output $\top$ if 
        $m=m_{i}$ and
        $\mathsf{OT}.\Ver(\pp,\sfot.\vk_{j-1},m_j\concat \sfot.\vk_j,\sfot.\sigma_j)=\top$ for every $j\in [i]$.
    \end{enumerate}

    \item 
    $\mathsf{MT}.\Del(\sigk_i)\to\cert:$ on input $\sigk$,
    proceed as follows:
    \begin{enumerate}
    \item Parse $\sigk_{i}$ as $(
    \sfot.\sigk_{i}, \{\sfot.\sigk'_{j}\}_{j\in \{0,1,...,i-1\}}, 
     \{m_j,\sfot.\vk_{j},\sfot.\sigma_j\}_{j\in [i]})$.
    \item For $j\in \{0,1,...,i-1\}$, run $\sfot.\cert_j \leftarrow \mathsf{OT}.\Del(\sfot.\sigk'_j)$.
    \item Run $\sfot.\cert_i \leftarrow \mathsf{OT}.\Del(\sfot.\sigk_i)$.
    \item Output $\cert:=\{\sfot.\cert_j,m_j,\sfot.\vk_{j},\sfot.\sigma_j\}_{j\in [i]}$.
    \end{enumerate}
    \item 
    $\mathsf{MT}.\Cert(\pp,\vk,\ck,\cert,S)\to\top/\bot:$ on input $\pp,\vk,\ck$, $\cert$, and $S$,  
    parse the certificate $\cert$ as a tuple $\{\sfot.\cert_j,m_j,\sfot.\vk_{j},\sfot.\sigma_j\}_{j\in [i]}$, let $\sfot.\vk_0 = \vk$, and 
    output $\top$ if the following holds:
    \begin{itemize}
    \item $S=\{m_1,m_2,...,m_{i}\}$,   
    \item $\mathsf{OT}.\Cert(\sfot.\vk_{j-1},\ck,\sfot.\cert_{j-1},\{m_{j}\concat \sfot.\vk_{j}\})=\top$ for every $j\in [i]$, and
    \item $\mathsf{OT}.\Cert(\sfot.\vk_i,\ck,\sfot.\cert_i,\emptyset)=\top$.
    \end{itemize}
\end{itemize}

We now prove the following theorem.

\begin{theorem}
Suppose that $(\mathsf{OT}.\setup,\mathsf{OT}.\KeyGen,\mathsf{OT}.\Sign,\mathsf{OT}.\Ver,\mathsf{OT}.\Del,\mathsf{OT}.\Cert)$ satisfies the one-time variants of correctness, EUF-CMA security, deletion correctness, and deletion security in \Cref{def:dswrsk}. Then, the scheme $(\mathsf{MT}.\setup,\mathsf{MT}.\KeyGen,\mathsf{MT}.\Sign,\mathsf{MT}.\Ver,\mathsf{MT}.\Del,\mathsf{MT}.\Cert)$ satisfies many-time variants of each of the properties.
\end{theorem}
\begin{proof}
Let $(\mathsf{OT}.\setup,\mathsf{OT}.\KeyGen,\mathsf{OT}.\Sign,\mathsf{OT}.\Ver,\mathsf{OT}.\Del,\mathsf{OT}.\Cert)$ be a scheme which satisfies the one-time variants in \Cref{def:dswrsk}. We now verify that each of the many-time properties are satisfied. 

\paragraph{Many-time correctness.}
Because the $\mathsf{OT}$ scheme satisfies one-time correctness, this means that there exists a negligible function $\nu$ such that for any message $m$ and for all $\lambda \in \N$,
\begin{align}\label{eq:ot-correctness}
\Pr\left[\top\gets \mathsf{OT}.\Ver(\vk,m,\sigma_1):
\begin{array}{rr}
(\pp,\ck)\gets \mathsf{OT}.\setup(1^\secp)\\
     (\sigk,\vk)\gets \mathsf{OT}.\KeyGen(\pp)  \\
     (\sigk_1,\sigma_1)\gets \mathsf{OT}.\Sign(\pp,\sigk,m)
\end{array}
\right]\ge1-\nu(\secp).    
\end{align}
Let $p=p(\lambda)$ be any polynomial and let $(m_1,m_2,...,m_{p})$ be any collection of messages. Then, it follows from the union bound and \Cref{eq:ot-correctness} that
\begin{align}\label{eq:many-time-correctness-MT}
&\Pr\left[\bigwedge_{i\in[p]}\top\gets \mathsf{MT}.\Ver(\vk,m_i,\sigma_i):
\begin{array}{rr}
(\pp,\ck)\gets\setup(1^\secp)\\
     (\sigk_0,\vk)\gets \mathsf{MT}.\KeyGen(\pp)  \\
     (\sigk_1,\sigma_1)\gets\mathsf{MT}.\Sign(\pp,\sigk_0,m_1)\\
     (\sigk_2,\sigma_2)\gets\mathsf{MT}.\Sign(\pp,\sigk_1,m_2)\\
     \vdots\\
     (\sigk_{p},\sigma_{p})\gets \mathsf{MT}.\Sign(\pp,\sigk_{p-1},m_{p})
\end{array}
\right] \\ 
&
\ge 1- p(\lambda) \cdot (p(\lambda)+1) \cdot \nu(\lambda)/2.
\end{align}
Finally, because $p(\lambda)=\poly(\lambda)$ and $\nu(\lambda) = \negl(\lambda)$ it follows that the above expression in \Cref{eq:many-time-correctness-MT} is at least $1-\negl(\lambda)$. This proves the claim.

\paragraph{EUF-CMA security.} The proof is virtually identical to the standard proof that the chain-based signature scheme achieves many-time EUF-CMA security, see~\cite{books/crc/KatzLindell2007}, for example. We choose to omit it for brevity.

\paragraph{Many-time deletion correctness.}
Because the $\mathsf{OT}$ scheme satisfies one-time correctness, this means that there exists a negligible function $\nu$ such that for any $p \le 1$ and any message $m_p$ (possibly $m_p = \emptyset$, if $p=0$, in which case no message is signed in the expression below) and for all $\lambda \in \N$,
\begin{align}\label{eq:ot-deletion-correctness}
\Pr\left[\top\gets\mathsf{OT}.\Cert(\pp,\vk,\ck,\cert,\{m_p\}):
\begin{array}{rr}
(\pp,\ck)\gets\mathsf{OT}.\setup(1^\secp)\\
     (\sigk_0,\vk)\gets\mathsf{OT}.\KeyGen(\pp)  \\
     (\sigk_p,\sigma_p)\gets\mathsf{OT}.\Sign(\pp,\sigk_0,m_p)\\
     \cert\gets \mathsf{OT}.\Del(\sigk_{p})
\end{array}
\right]\ge1-\nu(\secp). 
\end{align} 
Now let $p=p(\lambda)$ be any polynomial and let $(m_1,m_2,...,m_{p})$ be any collection of messages. In the case when $p =0$ the claim follows immediately from the one-time deletion correctness of the $\mathsf{OT}$ scheme. More generally in the case when $p(\lambda)\geq 1$ it follows from the union bound and \Cref{eq:ot-deletion-correctness} that
\begin{align}
&\Pr\left[\top\gets \mathsf{MT}.\Cert(\pp,\vk,\ck,\cert,\{m_1,m_2,...,m_p\}):
\begin{array}{rr}
(\pp,\ck)\gets\setup(1^\secp)\\
     (\sigk_0,\vk)\gets \mathsf{MT}.\KeyGen(\pp)  \\
     (\sigk_1,\sigma_1)\gets\mathsf{MT}.\Sign(\pp,\sigk_0,m_1)\\
     (\sigk_2,\sigma_2)\gets\mathsf{MT}.\Sign(\pp,\sigk_1,m_2)\\
     \vdots\\
     (\sigk_{p},\sigma_{p})\gets \mathsf{MT}.\Sign(\pp,\sigk_{p-1},m_{p})\\
     \cert\gets \mathsf{MT}.\Del(\sigk_{p})
\end{array}
\right] \\ 
&
\ge 1- (p(\lambda)+1) \cdot \nu(\lambda)/2.
\end{align}
Finally, because $p(\lambda)=\poly(\lambda)$ and $\nu(\lambda) = \negl(\lambda)$ it follows that the above expression above is at least $1-\negl(\lambda)$. This proves the claim.

\paragraph{Many-time deletion security.}
Suppose that there exists a QPT adversary $\cA$ which breaks the many-time deletion security of $\mathsf{MT}$. In other words, there exists a polynomial $q(\lambda)$ such that for infinitely many $\lambda \in \N$,
\begin{align}\label{eq:mt-deletion-security-adv}
   \Pr\left[
   \begin{array}{ll}
   &~~~\top\gets \mathsf{MT}.\Cert(\pp,\vk,\ck,\cert,S)\\
   &\wedge~ m^*\notin S\\
   &\wedge~ \top\gets\mathsf{MT}.\Ver(\vk,m^*,\sigma^*)
   \end{array}
   :
   \begin{array}{rr}
   (\pp,\ck)\gets \mathsf{MT}.\setup(1^\secp)\\
   (\vk,\cert,S,m^*,\sigma^*)\gets\cA(\pp)\\
   \end{array}
   \right] 
   \ge \frac{1}{q(\lambda)}.
\end{align}
We now show that we can use $\cA$ which satisfies \eqref{eq:mt-deletion-security-adv} to break the one-time deletion security (according to \Cref{def:dswrsk}) of the $\mathsf{OT}$ scheme. Consider the following reduction $\cB$:

\begin{enumerate}
    \item On input $\pp$, $\cB$ runs $\cA(\pp)$ and parses the output $(\vk,\cert,S,m^*,\sigma^*)$ as follows:
    \begin{itemize}
    \item $\vk=\sfot.\vk_0$

    \item $\cert=\{\sfot.\cert_j,m_j,\sfot.\vk_{j},\sfot.\sigma_j\}_{j\in [t]}$

    \item $\sigma^*=\{m^*_j,\sfot.\vk^*_{j},\sfot.\sigma^*_j\}_{j\in [t^*]}$
\end{itemize}
for some integers $t,t^*$. Let $i^* \leq \min\{t,t^*\}$ be the largest integer such that $(m_j,\sfot.\vk_i) = (m_j^*,\sfot.\vk_i^*)$ for all $i \in [i^*]$.

\item $\cB$ proceeds as follows (depending on which case it encounters):
\begin{itemize}
    \item[\textbf{Case 1:}] $i^*= t < t^*$. Let $\cert' = (\sfot.\cert_{i^*+1},m_{i^*+1},\sfot.\vk_{i^*+1},\sfot.\sigma_{i^*+1})$ and output $$\mathsf{out} = (\sfot.\vk_{i^*},\cert',\emptyset,m_{i^*+1}^*,\sigma_{i^*+1}^*).$$ 

    \item[\textbf{Case 2:}] $i^* < \min\{t,t^*\}$. Let $\cert' = (\sfot.\cert_{i^*+1},m_{i^*+1},\sfot.\vk_{i^*+1},\sfot.\sigma_{i^*+1})$ and output $$\mathsf{out} = (\sfot.\vk_{i^*},\cert',\{ m_{i^*+1}\},m_{i^*+1}^*,\sigma_{i^*+1}^*).$$

    \item[\textbf{Case 3:}] $(i^*= t^* < t)$ or $(i^*= t^*= t)$. Abort and output $\bot$.
\end{itemize}
\end{enumerate}
We now analyze the success probability of $\cB$. Note that, by our assumption in \eqref{eq:mt-deletion-security-adv}, either \textbf{Case 1} or \textbf{Case 2} must occur with probability at least $1/\poly(\lambda)$. Otherwise, if \textbf{Case 3} occurs with overwhelming probability, then $\cA$ does not constitute a successful adversary which contradicts our initial assumption in \eqref{eq:mt-deletion-security-adv}. 

Suppose that \textbf{Case 1} occurs. In this case, with probability at least $1/\poly(\lambda)$, we observe that $\mathsf{OT}.\Cert(\pp,\vk_{i^*},\ck,\cert',S')=\top$ and $m_{i^*+1}^* \notin S'$ for $S'=\emptyset$, and $\mathsf{OT}.\Ver(\vk_{i^*},m_{i^*+1}^*,\sigma_{i^*+1}^*) =\top$. 

Finally, suppose that \textbf{Case 2} occurs.
By the definition of $i^*$, it follows that $(m_{i^*+1},\sfot.\vk_{i^*+1}) \neq (m_{i^*+1}^*,\sfot.\vk_{i^*+1}^*)$. Hence, with probability at least $1/\poly(\lambda)$, we find that $\mathsf{OT}.\Cert(\pp,\vk_{i^*},\ck,\cert',S')=\top$ and $m_{i^*+1}^* \notin S'$ with respect to $S'=\{m_{i^*+1}\}$, as well as $\mathsf{OT}.\Ver(\vk_{i^*},m_{i^*+1}^*,\sigma_{i^*+1}^*) =\top$.

Therefore, we have shown that there exists a QPT adversary $\cB$ such that
\begin{align*}
   \Pr\left[
   \begin{array}{ll}
   &~~~\top\gets \mathsf{OT}.\Cert(\pp,\vk,\ck,\cert,S)\\
   &\wedge~ m^*\notin S\\
   &\wedge~ \top\gets\mathsf{OT}.\Ver(\vk,m^*,\sigma^*)\\
   &\wedge~|S|\leq 1
   \end{array}
   :
   \begin{array}{rr}
   (\pp,\ck)\gets \mathsf{OT}.\setup(1^\secp)\\
   (\vk,\cert,S,m^*,\sigma^*)\gets\cB(\pp)\\
   \end{array}
   \right] 
   \ge \frac{1}{\poly(\lambda)}.
\end{align*}
This proves the claim.

\end{proof}

\section{Adaptive Hardcore Bit Property for OWFs}
\label{sec:AHB}
In this section, we introduce a new concept, which we call {\it adaptive hardcore bit property for OWFs}, and show it from
the existence of OWFs. This property is inspired by the adaptive hardcore bit property which was shown for a family of noisy trapdoor claw-free functions by Brakerski et al.~\cite{JACM:BCMVV21}.
Our notion of the adaptive hardcore bit property for OWFs will be used to construct two-tier tokenized signatures.

\subsection{Statements}
The formal statement of the adaptive hardcore bit property for OWFs is given as follows.
(Its proof is given later.)
\begin{theorem}[Adaptive Hardcore Bit Property for OWFs]
\label{thm:adaptive_hardcore_bit}
Let $\lambda \in \N$ be the security parameter and let $\ell(\lambda),\kappa(\lambda) \in \N$ be polynomials. 
Let $f: \bit^{\ell(\secp)} \rightarrow \bit^{\kappa(\secp)}$ be a (quantumly-secure) OWF. Then, for any QPT adversary $\{\mathcal{A}_\secp\}_{\secp\in \N}$, it holds that
\begin{align}
\Pr\left[
\begin{array}{cc}
f(x) \in \{f(x_0),f(x_1)\} \vspace{1mm}\ \\
\bigwedge \vspace{1mm}\ \\
\, d \cdot (x_0 \oplus x_1) = c
\end{array}
:
\begin{array}{ll}
x_0 \leftarrow \bit^{\ell(\secp)}, \,\, x_1 \leftarrow \bit^{\ell(\secp)} \vspace{1.2mm}\\
c \leftarrow \bit\vspace{1mm}\\
(x,d) \leftarrow \mathcal{A}_\secp\left(\frac{\ket{x_0} +(-1)^c\ket{x_1}}{\sqrt{2}},f(x_0),f(x_1) \right)
\end{array}
\right] \leq \frac{1}{2} + \negl(\lambda).
\end{align}
\end{theorem}

We actually use its amplified version, which is given as follows.
(Its proof is given later.)
\begin{theorem}[Amplified Adaptive Hardcore Bit Property for OWFs]
\label{thm:adaptive_hardcore_bit_amp}
Let $\lambda \in \N$ be the security parameter and let  $\ell(\lambda),\kappa(\lambda),n(\secp) \in \N$ be polynomials. 
Let $f: \bit^{\ell(\secp)} \rightarrow \bit^{\kappa(\secp)}$ be a (quantumly-secure) OWF. Then, for any QPT adversary $\{\mathcal{A}_\secp\}_{\secp\in\N}$, it holds that,
\begin{align}
\Pr
\left[
\begin{array}{cc}
\wedge_{i \in [n]} f(x_i) \in \{f(x_i^0),f(x_i^1)\} \vspace{1mm}\ \\
\bigwedge \vspace{1mm}\ \\
\wedge_{i \in [n]} d_i \cdot (x_i^0 \oplus x_i^1) = c_i
\end{array}
:
\begin{array}{ll}
\forall i \in [n] \,: \, x_i^0\gets\bit^{\ell(\secp)}, \, x_i^1\gets\bit^{\ell(\secp)} \vspace{1mm}\\
\forall i \in [n] \,: \, c_i\gets\bit \vspace{1mm}\\
\{x_i,d_i\}_{i\in[n]} \leftarrow \mathcal{A}_\secp\left(\bigotimes_{i=1}^{n}\frac{\ket{x_i^0} +(-1)^{c_i}\ket{x_i^1}}{\sqrt{2}},\{f(x_i^b)\}_{i,b} \right)
\end{array}
\right] \leq \negl(\lambda).
\end{align}
\end{theorem}

\subsection{Theorem of \cite{bartusek2023weakening}}
In order to show adaptive hardcore bit property for OWFs,
we use the following theorem which is implicit in \cite[Theorem 3.1]{bartusek2023weakening}. The only difference is that we additionally reveal both pre-images as part of the distribution  $\left\{\widetilde{\cZ}_\secp^{\cA_\secp}(b)\right\}_{\secp \in \N, b \in \{0,1\}}$. We remark that the proof is the same. 

\begin{theorem}[Implicit in \cite{bartusek2023weakening}, Theorem 3.1]\label{thm:BKMPW}
Let $\secp\in\N$ be the security parameter,
and let $\ell(\secp),\kappa(\secp)\in\N$ be polynomials.
Let $f: \{0,1\}^{\ell(\secp)} \to \{0,1\}^{\kappa(\secp)}$ be a OWF secure against QPT adversaries. Let $\{\cZ_\secp(\cdot,\cdot,\cdot,\cdot)\}_{\secp \in \N}$ be a quantum operation with four arguments: an $\ell(\secp)$-bit string $z$, two $\kappa(\secp)$-bit strings $y_0,y_1$, and an $\ell(\secp)$-qubit quantum state $\ket{\psi}$. 
Suppose that for any QPT adversary $\{\cA_\secp\}_{\secp \in \N}$, $z \in \{0,1\}^{\ell(\secp)},y_0,y_1 \in \{0,1\}^{\kappa(\secp)}$, and $\ell(\secp)$-qubit state $\ket{\psi}$,

\[\bigg| \Pr\left[\cA_\secp(\cZ_\secp(z,y_0,y_1,\ket{\psi})) = 1\right] - \Pr\left[\cA_\secp(\cZ_\secp(0^{\ell(\secp)},y_0,y_1,\ket{\psi})) = 1\right]\bigg| = \negl(\secp).\]

That is, $\cZ_\secp$ is semantically-secure with respect to its first input. Now, for any QPT adversary $\{\cA_\secp\}_{\secp \in \N}$, consider the distribution $\left\{\widetilde{\cZ}_\secp^{\cA_\secp}(b)\right\}_{\secp \in \N,b\in\bit}$ over quantum states, obtained by running $\mathcal{A}_\lambda$ as follows.

\begin{itemize}
    \item Sample $x_0,x_1 \gets \{0,1\}^{\ell(\secp)}$, define $y_0 = f(x_0), y_1 = f(x_1)$ and initialize $\cA_\secp$ with 
    \[\cZ_\secp\left(x_0 \oplus x_1, y_0,y_1,\frac{\ket{x_0} + (-1)^b \ket{x_1}}{\sqrt{2}}\right).\]
    \item $\cA_\secp$'s output is parsed as a string $x' \in \{0,1\}^{\ell(\secp)}$ and a residual state on register $\mathsf{A}'$.
    \item If $f(x') \in \{y_0,y_1\}$, then output $(x_0,x_1,\mathsf{A}')$, and otherwise output $\bot$. 
\end{itemize}

Then, it holds that
\begin{align}
\left\| \widetilde{\cZ}_\secp^{\cA_\secp}(0) - \widetilde{\cZ}_\secp^{\cA_\secp}(1) \right\|_{\rm tr}  \leq \negl(\secp).
\end{align}
\end{theorem}

We can show the following parallel version. (It can be shown by the standard hybrid argument. A detailed proof is given in \cref{sec:BKMPW_parallel}.) 
\begin{theorem}[Parallel version of \cref{thm:BKMPW}]
\label{thm:BKMPW_amp}
Let $\secp\in\N$ be the security parameter.
Let $\ell(\secp),\kappa(\secp),n(\secp)\in\N$ be polynomials.
Let $f: \{0,1\}^{\ell(\secp)} \to \{0,1\}^{\kappa(\secp)}$ be a OWF secure against QPT adversaries. 
Let $\{\cZ_\secp(\cdot,\cdot,\cdot,\cdot)\}_{\secp \in \N}$ be a quantum operation with four arguments: 
an $\ell(\secp)$-bit string $z$, 
two $\kappa(\secp)$-bit strings $y_0,y_1$, 
and an $\ell(\secp)$-qubit quantum state $\ket{\psi}$. Suppose that for any QPT adversary 
$\{\cA_\secp\}_{\secp \in \N}$, 
$z \in \{0,1\}^{\ell(\secp)},y_0,y_1 \in \{0,1\}^{\kappa(\secp)}$, and $\ell(\secp)$-qubit state $\ket{\psi}$,

\[\bigg| \Pr\left[\cA_\secp(\cZ_\secp(z,y_0,y_1,\ket{\psi})) = 1\right] 
- \Pr\left[\cA_\secp(\cZ_\secp(0^{\ell(\secp)},y_0,y_1,\ket{\psi})) = 1\right]\bigg| = \negl(\secp).\]

That is, $\cZ_\secp$ is semantically-secure with respect to its first input. Now, for any QPT adversary $\{\cA_\secp\}_{\secp \in \N}$, 
consider the distribution $\left\{\widetilde{\cZ}_\secp^{\cA_\secp}(b_1,...,b_{n(\secp)}\right\}_{\secp \in \N,b_i\in\bit}$ over quantum states, obtained by running $\mathcal{A}_\lambda$ as follows.

\begin{itemize}
    \item Sample $x_i^0,x_i^1 \gets \{0,1\}^{\ell(\secp)}$ for each $i\in[n(\secp)]$, define $y_i^0 = f(x_i^0), y_i^1 = f(x_i^1)$ and initialize $\cA_\secp$ with 
   \begin{align}
    \bigotimes_{i\in[n(\secp)]}\cZ_\secp\left(x_i^0 \oplus x_i^1, y_i^0,y_i^1,
    \frac{\ket{x_i^0} + (-1)^{b_i} \ket{x_i^1}}{\sqrt{2}}\right).
   \end{align} 
    \item $\cA_\secp$'s output is parsed as strings $x'_i\in \{0,1\}^{\ell(\secp)}$ for $i\in[n(\secp)]$ and a residual state on register $\mathsf{A}'$.
    \item If $f(x'_i) \in \{y_i^0,y_i^1\}$ for all $i\in[n(\secp)]$, output $(\{x_i^0\}_{i\in[n(\secp)]},\{x_i^1\}_{i\in[n(\secp)]},\mathsf{A}')$, and otherwise output $\bot$.
\end{itemize}

Then, there exists a negligible function $\negl(\secp)$ such that for any $b_1,...,b_{n(\secp)}\in\bit$,
\begin{align}
\left\|\widetilde{\cZ}_\secp^{\cA_\secp}(b_1,...,b_{n(\secp)}) - \widetilde{\cZ}_\secp^{\cA_\secp}(0,...,0)\right\|_{\rm tr} \leq \negl(\secp).
\end{align}
\end{theorem}

\subsection{Proof of \cref{thm:adaptive_hardcore_bit}}
By using \cref{thm:BKMPW}, we can show \cref{thm:adaptive_hardcore_bit} as follows.
Here, we leverage the fact that any algorithm that simultaneously produces both a valid pre-image of the OWF, as well as a string which leaks information about the relative phase between the respective pre-images, must necessarily violate \cref{thm:BKMPW}.

\begin{proof}[Proof of \cref{thm:adaptive_hardcore_bit}]
Let $\ell(\lambda),\kappa(\lambda) \in \N$ be polynomials, and let $f: \{0,1\}^{\ell(\secp)} \to \{0,1\}^{\kappa(\secp)}$ be a (quantumly-secure) OWF. Suppose there exist a QPT algorithm $\{\mathcal{A}_\lambda\}_{\secp\in\N}$ and a polynomial $p(\lambda)$ such that, for random $x_0,x_1 \leftarrow \bit^\ell$ and $c \leftarrow \bit$, it holds that
\begin{align}
\Pr\left[
\substack{f(x) \in \{f(x_0),f(x_1)\}\ \\
\bigwedge\ \\
\, d \cdot (x_0 \oplus x_1) = c
}\, : \,
(x,d) \leftarrow \mathcal{A}_\secp\left(\frac{\ket{x_0} +(-1)^c\ket{x_1}}{\sqrt{2}},f(x_0),f(x_1) \right)
\right] \geq \frac{1}{2} + \frac{1}{p(\lambda)}
\label{assumeBKMPWisnotsecure}
\end{align} 
for infinitely many $\secp$.
We now show how to construct an algorithm that violates \Cref{thm:BKMPW}. For simplicity, we define the quantum operation $\{\cZ_\secp(\cdot,\cdot,\cdot,\cdot)\}_{\secp \in \N}$ in \Cref{thm:BKMPW} as 
\[\cZ_\secp\left(x_0 \oplus x_1, f(x_0),f(x_1),\frac{\ket{x_0} + (-1)^c \ket{x_1}}{\sqrt{2}}\right) \coloneqq \left(f(x_0),f(x_1),\frac{\ket{x_0} +(-1)^c\ket{x_1}}{\sqrt{2}} \right).\] 
Evidently, our choice of $\cZ_\secp$ is trivially semantically secure with respect to the first argument. 
Consider the following QPT algorithm $\mathcal{B}_\secp$:
\begin{enumerate}
    \item On input $\left(f(x_0),f(x_1),\frac{\ket{x_0} +(-1)^c\ket{x_1}}{\sqrt{2}} \right)$, run
    $$
    (x,d_c) \leftarrow \mathcal{A}_\secp\left(\frac{\ket{x_0} +(-1)^c\ket{x_1}}{\sqrt{2}},f(x_0),f(x_1) \right).
    $$
    \item Output $x$ and assign $\proj{d_c}$ as the residual state.\footnote{Note that we can think of $d_c$ as a classical mixture (i.e., density matrix) over the randomness of $x_0,x_1 \leftarrow \bit^\ell$, $c \leftarrow \bit$ and the internal randomness of the algorithm $\mathcal{A}_\secp$.}
\end{enumerate}
Adopting the notation from \Cref{thm:BKMPW}, 
we define $ \widetilde{\cZ}_\secp^{\cB_\secp}(c)$.\footnote{It is, roughly speaking, 
$ \proj{x_0} \otimes \proj{x_1} \otimes \proj{d_c}$ for $c \in \bit$
when $f(x)\in\{f(x_0),f(x_1)\}$, and is $\bot$ when $f(x)\notin\{f(x_0),f(x_1)\}$.}
Consider the following distinguisher that distinguishes $\widetilde{\cZ}_\secp^{\cB_\secp}(c)$ for $c\in\bit$:
\begin{enumerate}
    \item 
Get $\widetilde{\cZ}_\secp^{\cB_\secp}(c)$ as input.
\item 
If it is $\bot$, output $\bot$ and abort.
\item 
Output $d_c \cdot (x_0 \oplus x_1) \pmod{2}$. 
\end{enumerate}
From \cref{assumeBKMPWisnotsecure}, 
there exists a polynomial $p(\lambda)$ such that both $f(x)\in \{f(x_0),f(x_1)\}$ and $d_c \cdot (x_0 \oplus x_1) = c  \pmod{2}$ occur with probability at least $\frac{1}{2} + \frac{1}{p(\secp)}$. 
Thus, the distinguisher can distinguish $\widetilde{\cZ}_\secp^{\cB_\secp}(0)$ and $\widetilde{\cZ}_\secp^{\cB_\secp}(1)$ with probability at least $\frac{1}{2} + \frac{1}{p(\secp)}$, but this means
$$
\left\| \widetilde{\cZ}_\secp^{\cB_\secp}(0) - \widetilde{\cZ}_\secp^{\cB_\secp}(1) \right\|_{\rm tr} \geq \frac{2}{p(\lambda)}.
$$
from \Cref{thm:holevo-hesltrom}.
This violates \Cref{thm:BKMPW}.
\end{proof}

\subsection{Proof of \cref{thm:adaptive_hardcore_bit_amp}}
In this subsection, we show \cref{thm:adaptive_hardcore_bit_amp} by using \cref{thm:BKMPW_amp}.
\begin{proof}[Proof of \cref{thm:adaptive_hardcore_bit_amp}]
For the sake of contradiction,
assume that
there is a QPT adversary $\{\mathcal{A}_\secp\}_{\secp\in\N}$ such that
\begin{align}
\Pr
\left[
\begin{array}{cc}
\wedge_{ i \in [n]} f(x_i) \in \{f(x_i^0),f(x_i^1)\}\ \\
\bigwedge\ \\
\wedge_{i \in [n]} d_i \cdot (x_i^0 \oplus x_i^1) = c_i
\end{array}
:
\begin{array}{ll}
\forall i \in [n] \,: \, x_i^0\gets\bit^\ell, \, x_i^1\gets\bit^\ell, c_i\gets\bit\\
\{x_i,d_i\}_{i\in[n]} \leftarrow \mathcal{A}_\secp\left(\bigotimes_{i=1}^n\frac{\ket{x_i^0} +(-1)^{c_i}\ket{x_i^1}}{\sqrt{2}},\{f(x_i^b)\}_{i,b} \right)
\end{array}
\right] \ge \frac{1}{\poly(\secp)}
\label{assumption_parallel_hardcorebit_amp}
\end{align}
for infinitely many $\secp$.
We consider the quantum operation $\{\cZ_\secp(\cdot,\cdot,\cdot,\cdot)\}_{\secp \in \N}$ in \Cref{thm:BKMPW_amp} as 
\begin{align}
\cZ_\secp\left(x_0 \oplus x_1, f(x_0),f(x_1),\frac{\ket{x_0} + (-1)^c \ket{x_1}}{\sqrt{2}}\right) \coloneqq \left(f(x_0),f(x_1),\frac{\ket{x_0} +(-1)^c\ket{x_1}}{\sqrt{2}} \right),
\end{align}
which is trivially semantically secure with respect to the first argument. 
From such $\{\cA_\secp\}_{\secp\in\N}$ and $\{\cZ_\secp\}_{\secp\in\N}$, we construct the following QPT adversary $\{\cB_\secp\}_{\secp\in\N}$ for fixed each $(c_1,...,c_n)\in\bit^n$:
\begin{enumerate}
    \item 
    Get $\{f(x_i^b)\}_{i\in[n],b\in\bit}$ and $\bigotimes_{i\in[n]}\frac{\ket{x_i^0}+(-1)^{c_i}\ket{x_i^1}}{\sqrt{2}}$
    as input.
    \item
    Run $(\{x_i\}_{i\in[n]},\{d_i\}_{i\in[n]}) \gets \cA_\secp
    \left(\bigotimes_{i=1}^n\frac{\ket{x_i^0} +(-1)^{c_i}\ket{x_i^1}}{\sqrt{2}},\{f(x_i^b)\}_{i\in[n],b\in\bit} \right)
    $.
    \item 
    Output $\{x_i\}_{i\in[n]}$. Set its residual state as $\bigotimes_{i\in[n]}\proj{d_i}$.
\end{enumerate}
Then,
by using the notation of \cref{thm:BKMPW_amp},
we define $\widetilde{\cZ}_\secp^{\cB_\secp}(c_1,...,c_n)$.\footnote{Roughly speaking, it is
$\left(\bigotimes_{i\in[n],b\in\bit}\proj{x_i^b}\right)\otimes\left(\bigotimes_{i\in[n]}\proj{d_i}\right)$
if $f(x_i)\in\{f(x_i^0),f(x_i^1)\}$ for all $i\in[n]$,
and it is $\bot$ otherwise.}
Let us consider the following QPT distinguisher $\{\cD_\secp\}_{\secp\in\N}$:
\begin{enumerate}
    \item 
    Get $\widetilde{\cZ}_\secp^{\cB_\secp}(c_1,...,c_n)$ as input.
    \item 
   If it is $\bot$, output $\bot$.
   Otherwise, parse it as
$\left(\bigotimes_{i\in[n],b\in\bit}\proj{x_i^b}\right)\otimes\left(\bigotimes_{i\in[n]}\proj{d_i}\right)$.
\item 
Compute $c_i'\coloneqq d_i\cdot(x_i^0\oplus x_i^1)$ for each $i\in[n]$.
Output $\{c_i'\}_{i\in[n]}$. 
\end{enumerate}
Then, from \cref{assumption_parallel_hardcorebit_amp}, 
\begin{align}
\frac{1}{2^n}\sum_{(c_1,...,c_n)\in\bit^n}\Pr[(c_1,...,c_n)\gets\cD(
    \widetilde{\cZ}_\secp^{\cB_\secp}(c_1,...,c_n))]\ge\frac{1}{\poly(\secp)}
    \label{kiketu}
\end{align}
for infinitely many $\secp$.
Now we show that it contradicts \cref{thm:BKMPW_amp}.

If \cref{thm:BKMPW_amp} is correct,
there exists a negligible function $\negl$ such that
\begin{align}
    \left\|
    \widetilde{\cZ}_\secp^{\cB_\secp}(c_1,...,c_n)
    -\widetilde{\cZ}_\secp^{\cB_\secp}(0,...,0)
    \right\|_{\rm tr}\le\negl(\secp)
\end{align}
for all $(c_1,...,c_n)\in\bit^n$.
However, in that case, 
there exists a negligible function $\negl$
such that
\begin{align}
\left|
\Pr\left[(c_1,...,c_n)\gets\cD(\widetilde{\cZ}_\secp^{\cB_\secp}(c_1,...,c_n))\right]
-\Pr\left[(c_1,...,c_n)\gets\cD(\widetilde{\cZ}_\secp^{\cB_\secp}(0,...,0))\right]
\right|
\le\negl(\secp)
\label{neglsa}
\end{align}
for all
$(c_1,...,c_n)\in\bit^n$.
Then we have
\begin{align}
\frac{1}{\poly(\secp)}
&\le
\frac{1}{2^n}\sum_{(c_1,...,c_n)\in\bit^n}\Pr[(c_1,...,c_n)\gets\cD(
\widetilde{\cZ}_\secp^{\cB_\secp}(c_1,...,c_n))]\\
&\le
\frac{1}{2^n}\sum_{(c_1,...,c_n)\in\bit^n}\left(\Pr[(c_1,...,c_n)\gets\cD(
\widetilde{\cZ}_\secp^{\cB_\secp}(0,...,0))]+\negl(\secp)\right)\\
&\le
\frac{1}{2^n}\sum_{(c_1,...,c_n)\in\bit^n}\Pr[(c_1,...,c_n)\gets\cD(
\widetilde{\cZ}_\secp^{\cB_\secp}(0,...,0))]+\negl(\secp)\\
&\le
\frac{1}{2^n}
+\negl(\secp)
\end{align}
for infinitely many $\secp$,
which is the contradiction.
Here, the first inequality is from \cref{kiketu}, 
the second inequality is from \cref{neglsa},
and the last inequality is from the fact that
$\sum_{(c_1,...,c_n)\in\bit^n}\Pr[(c_1,...,c_n)\gets\cA]=1$ for any algorithm $\cA$.

\end{proof}

\section{Two-Tier Tokenized Signatures}
\label{sec:2-TS}

In this section, we will first give the formal definition of two-tier tokenized signatures (2-TS), and then show that they can be constructed from OWFs.
For the construction, we use the (amplified) adaptive hardcore bit property for OWFs (\cref{thm:adaptive_hardcore_bit_amp}).

\subsection{Definition}
The formal definition is as follows.
\begin{definition}[Two-Tier Tokenized Signatures (2-TS)]
A two-tier tokenized signature scheme is a tuple \break $(\KeyGen,\StateGen,\Sign,\Ver_0,\Ver_1)$ of algorithms such that
\begin{itemize}
    \item 
    $\KeyGen(1^\secp)\to (\sk,\pk):$
    It is a QPT algorithm that, on input the security parameter $\secp$, outputs
    a classical secret key $\sk$ and a classical public key $\pk$.
    \item 
    $\StateGen(\sk)\to\psi:$
    It is a QPT algorithm that, on input $\sk$, outputs a quantum state $\psi$.
    \item 
    $\Sign(\psi,m)\to\sigma:$
    It is a QPT algorithm that, on input $\psi$ and a message $m\in\bit$, outputs a classical signature $\sigma$.
    \item 
    $\Ver_0(\pk,\sigma)\to\top/\bot:$
    It is a QPT algorithm that, on input $\pk$ and $\sigma$, outputs $\top/\bot$.
    \item 
    $\Ver_1(\sk,\sigma)\to\top/\bot:$
    It is a QPT algorithm that, on input $\sk$ and $\sigma$, outputs $\top/\bot$.
\end{itemize}
We require the following properties.

\paragraph{Correctness:}
\begin{align}
\Pr\left[\top\gets\Ver_0(\pk,\sigma):
\begin{array}{rr}
(\sk,\pk)\gets\KeyGen(1^\secp)\\
\psi\gets\StateGen(\sk)\\
\sigma\gets\Sign(\psi,0)
\end{array}
\right]\ge1-\negl(\secp)    
\end{align}
and
\begin{align}
\Pr\left[\top\gets\Ver_1(\sk,\sigma):
\begin{array}{rr}
(\sk,\pk)\gets\KeyGen(1^\secp)\\
\psi\gets\StateGen(\sk)\\
\sigma\gets\Sign(\psi,1)
\end{array}
\right]\ge1-\negl(\secp).
\end{align}

\paragraph{Security:}
For any QPT adversary $\cA$,
\begin{align}
\Pr\left[
\top\gets\Ver_0(\pk,\sigma_0)
\wedge
\top\gets\Ver_1(\sk,\sigma_1)
:
\begin{array}{rr}
(\sk,\pk)\gets\KeyGen(1^\secp)\\
\psi\gets\StateGen(\sk)\\
(\sigma_0,\sigma_1)\gets\cA(\psi,\pk)
\end{array}
\right]\le\negl(\secp).
\end{align}
\end{definition}

We can show that the following type of security, which we call {\it one-wayness}, is also satisfied by
two-tier tokenized signatures.
\begin{lemma}[One-wayness of two-tier tokenized signatures]
\label{lem:onewayness}
For any QPT adversary $\cA$,
\begin{align}
\Pr\left[
\top\gets\Ver_0(\pk,\sigma)
:
\begin{array}{rr}
(\sk,\pk)\gets\KeyGen(1^\secp)\\
\psi\gets\cA(\pk)\\
\sigma\gets\Sign(\psi,0)
\end{array}
\right]\le\negl(\secp).
\end{align}
\end{lemma}

\begin{proof}
Assume that there exists a QPT adversary $\cA$ such that
\begin{align}
\Pr\left[
\top\gets\Ver_0(\pk,\sigma)
:
\begin{array}{rr}
(\sk,\pk)\gets\KeyGen(1^\secp)\\
\psi\gets\cA(\pk)\\
\sigma\gets\Sign(\psi,0)
\end{array}
\right]\ge\frac{1}{\poly(\secp)}
\end{align}
for infinitely many $\secp$.
Then, from such $\cA$, we can construct a QPT adversary $\cB$ that breaks the security of the two-tier tokenized signature scheme as follows:
\begin{enumerate}
    \item 
    Get $\psi$ and $\pk$ as input.
    \item 
    Run $\psi'\gets\cA(\pk)$.
    \item 
    Run $\sigma_0\gets\Sign(\psi',0)$ and
    $\sigma_1\gets\Sign(\psi,1)$.
    \item 
    Output $(\sigma_0,\sigma_1)$.
\end{enumerate}
It is clear that $\cB$ breaks the security of the two-tier tokenized signature scheme.
\end{proof}

\subsection{Construction}
We show that 2-TS can be constructed from OWFs.
\begin{theorem}
If OWFs exist, then two-tier tokenized signatures exist.
\end{theorem}

\begin{proof}
Let $f$ be a OWF. From it, we construct a two-tier tokenized signature scheme as follows:
\begin{itemize}
     \item 
    $\KeyGen(1^\secp)\to (\sk,\pk):$
    Choose $x_i^0,x_i^1\gets\bit^\ell$ for each $i\in[n]$.
    Choose $c_i\gets\bit$ for each $i\in[n]$.
    Output $\sk\coloneqq\{c_i,x_i^0,x_i^1\}_{i\in[n]}$
    and $\pk\coloneqq\{f(x_i^0),f(x_i^1)\}_{i\in[n]}$.
    \item 
    $\StateGen(\sk)\to\psi:$
    Parse $\sk=\{c_i,x_i^0,x_i^1\}_{i\in[n]}$.
    Output $\psi\coloneqq\bigotimes_{i\in[n]}\frac{\ket{x_i^0}+(-1)^{c_i}\ket{x_i^1}}{\sqrt{2}}$. 
    \item 
    $\Sign(\psi,m)\to\sigma:$
    If $m=0$, measure $\psi$ in the computational basis to get the result $\{z_i\}_{i\in[n]}$ (where $z_i\in\bit^\ell$ for each $i\in[n]$), and output it as $\sigma$.
    If $m=1$, measure $\psi$ in the Hadamard basis to get the result $\{d_i\}_{i\in[n]}$ (where $d_i\in\bit^\ell$ for each $i\in[n]$), and output it as $\sigma$.
    \item 
    $\Ver_0(\pk,\sigma)\to\top/\bot:$
    Parse 
    $\pk=\{f(x_i^0),f(x_i^1)\}_{i\in[n]}$ and $\sigma=\{z_i\}_{i\in[n]}$. 
    If $f(z_i)\in\{f(x_i^0),f(x_i^1)\}$ for all $i\in[n]$, output $\top$.
    Otherwise, output $\bot$.
    \item 
    $\Ver_1(\sk,\sigma)\to\top/\bot:$
   Parse $\sk=\{c_i,x_i^0,x_i^1\}_{i\in[n]}$ and $\sigma=\{d_i\}_{i\in[n]}$.
    If $d_i\cdot(x_i^0\oplus x_i^1)=c_i$ for all $i\in[n]$, output $\top$. 
    Otherwise, output $\bot$.
\end{itemize}
The correctness is clear.
The security is also clear from \cref{thm:adaptive_hardcore_bit_amp}.
\end{proof}

\section{Digital Signatures with Revocable Signatures}
\label{sec:sigcert}
In this section, we define digital signatures with revocable signatures ($\mathsf{DSR\mbox{-}Sign}$).
We also show that it can be constructed from 2-TS,
and therefore from OWFs.

\subsection{Definition}
We first give its formal definition as follows.
\begin{definition}[Digital Signatures with Revocable Signatures ($\mathsf{DSR\mbox{-}Sign}$)]
A digital signature scheme with revocable signatures is a set $(\KeyGen,\Sign,\Ver,\Del,\Cert)$    
of algorithms that satisfy the following.
\begin{itemize}
    \item 
    $\KeyGen(1^\secp)\to (\sigk,\vk):$
    It is a QPT algorithm that, on input the security parameter $\secp$, outputs a classical signing key $\sigk$ and
    a classical public verification key $\vk$. 
    \item 
    $\Sign(\sigk,m)\to(\psi,\ck):$
    It is a QPT algorithm that, on input a message $m$ and $\sigk$, outputs a quantum signature $\psi$
    and a classical check key $\ck$.
    \item 
    $\Ver(\vk,\psi,m)\to\top/\bot:$
    It is a QPT algorithm that, on input $\vk$, $m$, and $\psi$, outputs $\top/\bot$.
    \item 
    $\Del(\psi)\to\cert:$
    It is a QPT algorithm that, on input $\psi$, outputs a classical certificate $\cert$.
    \item 
    $\Cert(\ck,\cert)\to\top/\bot:$
    It is a QPT algorithm that, on input $\ck$ and $\cert$, outputs $\top/\bot$.
\end{itemize}
We require the following properties.

\paragraph{Correctness:}
For any message $m$,
\begin{align}
\Pr
\left[\top\gets\Ver(\vk,\psi,m):
\begin{array}{rr}
(\sigk,\vk)\gets\KeyGen(1^\secp)\\
(\psi,\ck)\gets\Sign(\sigk,m)
\end{array}
\right]
\ge1-\negl(\secp).
\end{align}

\paragraph{Deletion correctness:}
For any message $m$,
\begin{align}
   \Pr\left[
   \top\gets\Cert(\ck,\cert):
   \begin{array}{rr}
  (\sigk,\vk)\gets\KeyGen(1^\secp)\\
  (\psi,\ck)\gets\Sign(\sigk,m)\\
   \cert\gets\Del(\psi)
  \end{array}
  \right] 
   \ge1-\negl(\secp).
\end{align}

\paragraph{Many-time deletion security:}
For any adversary $\cA$ consisting of a pair of QPT algorithms $(\cA_1,\cA_2)$:
\begin{align}
\Pr\left[\top\gets\Cert(\ck^*,\cert)\wedge \top\gets\Ver(\vk,\psi,m^*):
   \begin{array}{rr}
   (\sigk,\vk)\gets\KeyGen(1^\secp)\\
   (m^*,\st)\gets\cA_1^{\Sign(\sigk,\cdot)}(\vk)\\
    (\psi^*,\ck^*)\gets\Sign(\sigk,m^*)\\
    (\cert,\psi)\gets\cA_2^{\Sign(\sigk,\cdot)}(\st,\psi^*)\\
   \end{array}
   \right] 
   \le\negl(\secp),
   \label{manytimedeletionsecurity}
\end{align}
where $\cA$ is not allowed to query $m^*$ to the signing oracle.
\end{definition}

\begin{remark}
The above definition does not capture the situation where the adversary gets more than one signatures on $m^*$ but deleted all of them. Actually, our construction seems to also satisfy security in such a setting. However, we choose to not formalize it for simplicity. 
\end{remark}
\begin{remark}
We can define the standard EUF-CMA security as follows, but it is trivially implied by many-time deletion security,
and therefore we do not include EUF-CMA security in the definition of digital signatures with revocable signatures. 
\begin{definition}[EUF-CMA Security]
For any QPT adversary $\cA$,
\begin{align}
\Pr\left[\top\gets\Ver(\vk,\psi^*,m^*):
\begin{array}{rr}
(\sigk,\vk)\gets\KeyGen(1^\secp)\\
(m^*,\psi^*)\gets\cA^{\Sign(\sigk,\cdot)}(\vk)
\end{array}
\right]    
\le\negl(\secp),
\label{EUF-CMA}
\end{align}
where $\cA$ is not allowed to query $m^*$ to the signing oracle.
\end{definition}
\end{remark}

We define a weaker version of many-time deletion security, which we call no-query deletion security as follows.
\begin{definition}[No-Query Deletion Security]
It is the same as many-time deletion security, \cref{manytimedeletionsecurity}, except that 
$\cA$ cannot query the signing oracle.

\end{definition}

The no-query security notion actually implies the many-time case:
\begin{lemma}[Many-Time Deletion Security from No-Query Deletion Security]\label{lem:extend}
Assume that EUF-CMA secure digital signature schemes exist. 
Then following holds: 
if there exists a digital signature scheme with revocable signatures which satisfies no-query deletion security, 
then there is a scheme that satisfies many-time deletion security.
\end{lemma}
\begin{proof}
Let $(\mathsf{NQ}.\KeyGen,\mathsf{NQ}.\Sign,\mathsf{NQ}.\Ver,\mathsf{NQ}.\Del,\mathsf{NQ}.\Cert)$ be a digital signature scheme with
revocable signatures that satisfies no-query deletion security.
Let $(\mathsf{MT}.\KeyGen,\mathsf{MT}.\Sign,\mathsf{MT}.\Ver)$ be a plain EUF-CMA secure digital signature scheme. 
From them, we can construct a digital signature scheme 
$\Sigma \coloneqq (\KeyGen,\Sign,\Ver,\Del,\Cert)$ 
with revocable signatures that satisfies
many-time deletion security 
as follows.
\begin{itemize}
    \item 
    $\KeyGen(1^\secp)\to (\sigk,\vk):$
    Run $(\mathsf{mt}.\sigk,\mathsf{mt}.\vk)\gets\mathsf{MT}.\KeyGen(1^\secp)$.
    Output $\sigk\coloneqq\mathsf{mt}.\sigk$ and $\vk\coloneqq\mathsf{mt}.\vk$.
    \item 
    $\Sign(\sigk,m)\to (\psi,\mathsf{ck}):$
    Parse $\sigk=\mathsf{mt}.\sigk$.
    Run $(\mathsf{nq}.\sigk,\mathsf{nq}.\vk)\gets\mathsf{NQ}.\KeyGen(1^\secp)$.
    Run $(\phi,\mathsf{nq}.\ck)\gets\mathsf{NQ}.\Sign(\mathsf{nq}.\sigk,m)$.
    Run $\sigma \gets \mathsf{MT}.\Sign(\mathsf{mt}.\sigk,\mathsf{nq}.\vk\|m)$. 
    Output $\psi\coloneqq(\phi,\sigma,\mathsf{nq}.\vk)$ 
    and $\mathsf{ck} \coloneqq \mathsf{nq.ck}$.
    \item 
    $\Ver(\vk,\psi,m)\to\top/\bot:$
    Parse $\vk=\mathsf{mt}.\vk$ and
    $\psi=(\phi,\sigma,\mathsf{nq}.\vk)$.
    Run $\mathsf{MT}.\Ver(\mathsf{mt}.\vk,\sigma,\mathsf{nq}.\vk\|m)$.
    If the output is $\bot$, output $\bot$ and abort.
    Run $\mathsf{NQ}.\Ver(\mathsf{nq}.\vk,\phi,m)$.
    If the output is $\top$, output $\top$. Otherwise, output $\bot$.
    \item 
    $\Del(\psi)\to\cert:$
    Parse $\psi=(\phi,\sigma,\mathsf{nq}.\vk)$.
    Run $\cert'\gets\mathsf{NQ}.\Del(\phi)$.
    Output $\cert\coloneqq\cert'$.
    \item 
    $\Cert(\mathsf{ck},\cert)\to\top/\bot:$
    Parse $\mathsf{ck} = \mathsf{nq}.\ck$. 
    Run $\mathsf{NQ}.\Cert(\mathsf{nq}.\ck,\cert)$, and output its output.
\end{itemize}

We show that 
$\Sigma$ satisfies many-time deletion security.
In other words, we show that
if the many-time deletion security of $\Sigma$ is broken,
then either the no-query deletion security of the digital signature scheme $\mathsf{NQ}$ is broken or
the EUF-CMA security of the digital signature scheme $\mathsf{MT}$ is broken.
Assume that
there exists a pair of QPT algorithms $\cA\coloneqq(\cA_1,\cA_2)$ such that
\begin{align}
\Pr\left[\top\gets\Cert(\ck^*,\cert)\wedge \top\gets\Ver(\vk,\psi,m^*):
   \begin{array}{rr}
   (\sigk,\vk)\gets\KeyGen(1^\secp)\\
   (m^*,\st)\gets\cA_1^{\Sign(\sigk,\cdot)}(\vk)\\
    (\psi^*,\ck^*)\gets\Sign(\sigk,m^*)\\
    (\cert,\psi)\gets\cA_2^{\Sign(\sigk,\cdot)}(\st,\psi^*)\\
   \end{array}
   \right] 
   \ge\frac{1}{\poly(\secp)}
   \label{manytimedeletionsecurity_break}
\end{align}
for infinitely many $\secp$,
where $\cA$ is not allowed to query $m^*$ to the signing oracle.
From such $\cA$, we construct a QPT adversary $\cB$ that breaks the no-query deletion security of the scheme $\mathsf{NQ}$ as follows:
    Let $\cC$ be the challenger of the security game of the no-query deletion security.
\begin{enumerate}
    \item 
    $\cC$ runs $(\mathsf{nq}.\sigk^*,\mathsf{nq}.\vk^*)\gets\mathsf{NQ}.\KeyGen(1^\secp)$.
    \item 
    $\cC$ sends $\mathsf{nq}.\vk^*$ to $\cB$.
    \item 
    $\cB$ runs $(\mathsf{mt}.\sigk,\mathsf{mt}.\vk)\gets\mathsf{MT}.\KeyGen(1^\secp)$.
    \item 
    $\cB$ runs $(m^*,\st)\gets\cA_1^{\Sign(\sigk,\cdot)}(\mathsf{mt}.\vk)$.
    When $\cA_1$ queries $m$ to the signing oracle, $\cB$ simulates it as follows:
    \begin{enumerate}
        \item 
        Run $(\mathsf{nq}.\sigk,\mathsf{nq}.\vk)\gets\mathsf{NQ}.\KeyGen(1^\secp)$.
        \item 
        Run $(\phi,\mathsf{nq}.\ck)\gets\mathsf{NQ}.\Sign(\mathsf{nq}.\sigk,m)$.
        \item 
        Run $\sigma\gets\mathsf{MT}.\Sign(\mathsf{mt}.\sigk,\mathsf{nq}.\vk\|m)$.
        \item 
        Output $\psi\coloneqq(\phi,\sigma,\mathsf{nq}.\vk)$ and $\ck\coloneqq \mathsf{nq}.\ck$.
    \end{enumerate}
    \item 
    $\cB$ sends $m^*$ to $\cC$.
    \item 
    $\cC$ runs $(\phi^*,\mathsf{nq}.\ck^*)\gets\mathsf{NQ}.\Sign(\mathsf{nq}.\sigk^*,m^*)$, and sends $\phi^*$ to $\cB$.
   \item 
    $\cB$ runs $\sigma^*\gets\mathsf{MT}.\Sign(\mathsf{mt}.\sigk,\mathsf{nq}.\vk^*\|m^*)$.
    \item 
    $\cB$ runs $(\cert,\psi)\gets\cA_2^{\Sign(\sigk,\cdot)}((\phi^*,\sigma^*,\mathsf{nq}.\vk^*))$.
        When $\cA_2$ queries $m$ to the signing oracle, $\cB$ simulates it as follows:
    \begin{enumerate}
        \item 
        Run $(\mathsf{nq}.\sigk,\mathsf{nq}.\vk)\gets\mathsf{NQ}.\KeyGen(1^\secp)$.
        \item 
        Run $(\phi,\mathsf{nq}.\ck)\gets\mathsf{NQ}.\Sign(\mathsf{nq}.\sigk,m)$.
        \item 
        Run $\sigma\gets\mathsf{MT}.\Sign(\mathsf{mt}.\sigk,\mathsf{nq}.\vk\|m)$.
        \item 
        Output $\psi\coloneqq(\phi,\sigma,\mathsf{nq}.\vk)$ and $\ck\coloneqq \mathsf{nq}.\ck$.
    \end{enumerate}
    \item
    Parse $\psi=(\phi,\sigma,\eta)$.
    $\cB$ outputs $\cert$ and $\phi$.
\end{enumerate}
Due to the EUF-CMA security of the scheme $\mathsf{MT}$, 
$\top\gets\mathsf{MT}.\Ver(\mathsf{mt}.\vk,\sigma,\eta\|m^*)$
occurs only when $\eta=\mathsf{nq}.\vk^*$ except for a negligible probability.
Therefore, \cref{manytimedeletionsecurity_break} means that both
$\Pr[\top\gets\mathsf{NQ}.\Ver(\mathsf{nq}.\vk^*,\phi,m^*)]$
and
$\Pr[\top\gets\mathsf{NQ}.\Cert(\mathsf{nq}.\ck^*,\cert)]$
are non-negligible for the above $\cB$,
which breaks the no-query deletion security of the scheme $\mathsf{NQ}$.

\end{proof}

\subsection{Construction}
Here we show the following result.
\begin{theorem}
\label{thm:revocablesignatures}
If two-tier tokenized signatures exist, then digital signatures with revocable signatures that satisfy
no-query deletion security exist.    
\end{theorem}

From \cref{lem:extend}, it also means the following:
\begin{corollary}
Digital signatures with revocable signatures (that satisfy many-time deletion security)
exist if two-tier tokenized signatures and EUF-CMA secure digital signatures exist.
\end{corollary}

\begin{proof}[Proof of \cref{thm:revocablesignatures}]
Here, we construct the scheme for the single-bit message space. 
It is clear that this can be extended to any fixed multi-bit message space case by the parallel execution of the protocol.
Moreover, by using universal one-way hash functions, 
it can be extended to unbounded poly-length message space case~\cite{STOC:NaoYun89}.

Let $(\mathsf{TS}.\KeyGen,\mathsf{TS}.\StateGen,\mathsf{TS}.\Sign,\mathsf{TS}.\Ver_0,\mathsf{TS}.\Ver_1)$ be a two-tier tokenized signature scheme.
From it, we construct a digital signature scheme with revocable signatures $\Sigma\coloneqq(\KeyGen,\Sign,\Ver,\Del,\Cert)$ 
that satisfies no-query deletion security
for the single bit message space as follows.
\begin{itemize}
    \item 
    $\KeyGen(1^\secp)\to(\sigk,\vk):$
    Run $(\sk_0,\pk_0)\gets\mathsf{TS}.\KeyGen(1^\secp)$.
    Run $(\sk_1,\pk_1)\gets\mathsf{TS}.\KeyGen(1^\secp)$.
    Output $\sigk\coloneqq(\sk_0,\sk_1)$ and $\vk\coloneqq(\pk_0,\pk_1)$.
    \item 
    $\Sign(\sigk,m)\to(\psi,\ck):$
    Parse $\sigk=(\sk_0,\sk_1)$.
    Run $\psi'\gets\mathsf{TS}.\StateGen(\sk_m)$.
    Output $\psi\coloneqq\psi'$ and $\ck\coloneqq \sk_m$.
    \item 
    $\Ver(\vk,\psi,m)\to\top/\bot:$
    Parse $\vk\coloneqq(\pk_0,\pk_1)$.
    Run $\sigma\gets\mathsf{TS}.\Sign(\psi,0)$.
    Run $\mathsf{TS}.\Ver_0(\pk_m,\sigma)$, and output its output.\footnote{The verification algorithm destroys the signature, but it can be done in a non-destructive way by 
    coherently applying this procedure and then doing the uncomputation.}

    \item 
    $\Del(\psi)\to\cert:$
    Run $\sigma\gets\mathsf{TS}.\Sign(\psi,1)$, and output $\cert\coloneqq \sigma$.
    \item 
    $\Cert(\ck,\cert)\to\top/\bot:$
    Parse $\ck=\sk_m$.
    Run $\mathsf{TS}.\Ver_1(\sk_m,\cert)$, and output its output.
\end{itemize}
Correctness and the deletion correctness are clear.
Let us show the no-query deletion security. Assume that 
there is a pair of QPT algorithms $(\cA_1,\cA_2)$  such that
\begin{align}
\Pr\left[\top\gets\Cert(\ck^*,\cert)\wedge \top\gets\Ver(\vk,\psi,m^*):
   \begin{array}{rr}
   (\sigk,\vk)\gets\KeyGen(1^\secp)\\
   (m^*,\st)\gets\cA_1(\vk)\\
    (\psi^*,\ck^*)\gets\Sign(\sigk,m^*)\\
    (\cert,\psi)\gets\cA_2(\st,\psi^*)\\
   \end{array}
   \right] 
   \ge\frac{1}{\poly(\secp)}
   \label{noquerydeletionsecurity}
\end{align}
for infinitely many $\secp$.
From $\cA$, we can construct a QPT adversary $\cB$ that breaks the original two-tier tokenized signature scheme as follows:
\begin{enumerate}
    \item 
    Get $\psi^*$ and $\pk$ as input.
    \item 
    Run $(\sk',\pk')\gets \mathsf{TS}.\KeyGen(1^\secp)$.
    Choose $r\gets\bit$.
    If $r=0$, set $\vk\coloneqq(\pk,\pk')$.
    If $r=1$, set $\vk\coloneqq(\pk',\pk)$.
    \item
   Run $(m^*,\st)\gets\cA_1(\vk)$. 
   If $r\neq m^*$, output $\bot$ and abort.
   \item 
   Run $(\cert,\psi)\gets\cA_2(\mathsf{st},\psi^*)$.
    \item 
    Run $\sigma_0\gets\mathsf{TS}.\Sign(\psi,0)$.
    Define $\sigma_1\coloneqq \cert$.
    \item 
    Output $(\sigma_0,\sigma_1)$.
\end{enumerate}
It is clear that $\Pr[\mbox{$\cB$ breaks the two-tier tokenized signature scheme}]\ge
\frac{1}{2}\Pr[\mbox{$\cA$ breaks $\Sigma$}]$.
Therefore, 
from \cref{noquerydeletionsecurity},
$\cB$ breaks the two-tier tokenized signature scheme.
\end{proof}

\ifnum\anonymous=1
\else
\paragraph{Acknowledgments.}
TM is supported by
JST CREST JPMJCR23I3,
JST Moonshot R\verb|&|D JPMJMS2061-5-1-1, 
JST FOREST, 
MEXT QLEAP, 
the Grant-in-Aid for Scientific Research (B) No.JP19H04066, 
the Grant-in Aid for Transformative Research Areas (A) 21H05183,
and 
the Grant-in-Aid for Scientific Research (A) No.22H00522.
\fi

\ifnum\submission=0
\bibliographystyle{alpha} 
\else
\bibliographystyle{splncs04}
\fi
\bibliography{abbrev3,crypto,reference}

\ifnum\submission=1
\else
\appendix 
\section{Proof of \cref{thm:BKMPW_amp}}
\label{sec:BKMPW_parallel}
Here we show \cref{thm:BKMPW_amp}.
Assume that there exist a QPT adversary $\cA_\secp$, a polynomial $\poly$, and
bits $b_1,...,b_n\in\bit$ such that
\begin{align}
\left\|\widetilde{\cZ}_\secp^{\cA_\secp}(b_1,...,b_n)-\widetilde{\cZ}_\secp^{\cA_\secp}(0,...,0)\right\|_{\rm tr} \ge 
\frac{1}{\poly(\secp)}
\end{align}
for infinitely often $\secp$.
Then due to the triangle inequality of the trace distance, there exist an integer $k$ and bits
$a_1,...,a_{k-1},a_{k+1},...,a_n\in\bit$ such that
\begin{align}
\left\|\widetilde{\cZ}_\secp^{\cA_\secp}(a_1,...,a_{k-1},0,a_{k+1},...,a_n)-
\widetilde{\cZ}_\secp^{\cA_\secp}(a_1,...,a_{k-1},1,a_{k+1},...,a_n)\right\|_{\rm tr} \ge 
\frac{1}{\poly(\secp)}.
\label{triangle_td}
\end{align}
From such $\cA_\secp$, we can construct a QPT adversary $\cB_\secp$ that breaks \cref{thm:BKMPW} as follows:
\begin{enumerate}
    \item 
    Get 
    \begin{align}
    \cZ_\secp\left(x_0 \oplus x_1, y_0,y_1,\frac{\ket{x_0} + (-1)^b \ket{x_1}}{\sqrt{2}}\right)
    \end{align}
        as input.
    \item Sample $x_i^0,x_i^1 \gets \{0,1\}^{\ell(\secp)}$ for each $i\in[n]\setminus\{k\}$. 
    Define $y_i^0 = f(x_i^0)$ and $y_i^1 = f(x_i^1)$ for each $i\in[n]\setminus\{k\}$.
    Initialize $\cA_\secp$ with 
   \begin{align}
    &\bigotimes_{i=1}^{k-1}\cZ_\secp\left(x_i^0 \oplus x_i^1, y_i^0,y_i^1, \frac{\ket{x_i^0} + (-1)^{a_i} \ket{x_i^1}}{\sqrt{2}}\right)\\
    &\otimes
    \cZ_\secp\left(x_0 \oplus x_1, y_0,y_1, \frac{\ket{x_0} + (-1)^{b} \ket{x_1}}{\sqrt{2}}\right)\\
    &\otimes
    \bigotimes_{i=k+1}^n\cZ_\secp\left(x_i^0 \oplus x_i^1, y_i^0,y_i^1, \frac{\ket{x_i^0} + (-1)^{a_i} \ket{x_i^1}}{\sqrt{2}}\right).
   \end{align} 
    \item $\cA_\secp$'s output is parsed as strings $x'_i\in \{0,1\}^{\ell(\secp)}$ for $i\in[n]$ and a residual state on register $\mathsf{A}'$.
    \item If $f(x'_i) \in \{y_i^0,y_i^1\}$ for all $i\in[n]\setminus\{k\}$, then output $x_k'$ and $\mathsf{A}'$. Otherwise output $\bot$.
\end{enumerate}
It is clear that the output of the above game for $b\in\bit$ is equivalent to
$\widetilde{\cZ}_\secp^{\cA_\secp}(a_1,...,a_{k-1},b,a_{k+1},...,a_n)$ (up to $\{x_i^0,x_i^1\}_{i\in[n]\setminus\{k\}}$).
Therefore, from \cref{triangle_td},
$\cB_\secp$ breaks \cref{thm:BKMPW}.

\section{Digital Signatures with Quantum Revocable Signing Keys from Group Actions} 
\label{sec:GA}
In this section, we construct digital signatures with {\it quantum} revocable signing keys from
group actions. Note that in this case, the revocation is quantum, i.e., not a classical deletion certificate but
the quantum signing key itself is returned.
The syntax, correctness, and security are the same as those of the classical revocation case, \cref{def:dswrsk}, except that
there is no deletion algorithm, $\Del$, and $\cert$ is not a classical bit string but the signing key $\sigk_i$ itself.

The outline of this section is as follows.
In \cref{sec:GAbasic}, we review basics of group actions.
We then review claw-free swap-trapdoor function pairs~\cite{EC:HhaMorYam23} and mention that it can be constructed from group actions with
the one-wayness property in \cref{sec:STF}.
We next define two-tier quantum lightning with classical semi-verification in \cref{sec:2QLclassical},
and construct it from claw-free swap-trapdoor function pairs (and therefore from group actions with the one-wayness property) in \cref{sec:GAtoQL}.
We finally construct digital signatures with quantum revocable signing keys from
two-tier quantum lightning with classical semi-verfication in \cref{sec:2QLtoDS}.

\subsection{Group Actions}
\label{sec:GAbasic}
Here we review basics of group actions~\cite{TCC:JQSY19}. 
Content of this subsection is taken from \cite{EC:HhaMorYam23}.
First, group actions are defined as follows.
\begin{definition}[Group actions]
Let $G$ be a (not necessarily abelian) group, $S$ be a set, and $\star:G\times S \rightarrow S$ be a function where we write $g\star s$ to mean $\star(g,s)$. We say that $(G,S,\star)$ is a \emph{group action} if it satisfies the following:
\begin{enumerate}
    \item For the identity element $e\in G$ and any $s\in S$, we have $e\star s=s$.
    \item For any $g,h\in G$ and any $s\in S$, we have $(gh)\star s = g\star(h\star s)$.
\end{enumerate}
\end{definition}

To be useful for cryptography, we have to at least assume that basic operations about $(G,S,\star)$ have efficient algorithms. We require the following efficient algorithms similarly to \cite{TCC:JQSY19}.
\begin{definition}[Group actions with efficient algorithms]\label{def:GA_efficient}
We say that a group action $(G,S,\star)$ has \emph{efficient algorithms} if it satisfies the following:\footnote{Strictly speaking, we have to consider a family $\{(G_\secp ,S_\secp ,\star_\secp)\}_{\secp \in \mathbb{N}}$ of group actions parameterized by the security parameter to meaningfully define the efficiency requirements. We omit the dependence on $\secp$ for notational simplicity throughout the paper.} 
\begin{itemize}
    \item[{\bf Unique representations:}] Each element of $G$ and $S$ can be represented as a bit string of length $\poly(\secp)$ in a unique manner. Thus, we identify these elements and their representations.
    \item[{\bf Group operations:}] There are classical deterministic polynomial-time algorithms that compute $gh$ from $g\in G$ and $h \in G$ and $g^{-1}$ from $g \in G$.
    \item[{\bf Group action:}] There is a classical deterministic polynomial-time algorithm that computes $g\star s$ from $g\in G$ and $s\in S$.
    \item[{\bf Efficient recognizability:}] There are classical deterministic polynomial-time algorithms that decide if a given bit string represents an element of $G$ or $S$, respectively.
    \item[{\bf Random sampling:}] There are PPT algorithms that sample almost uniform elements of $G$ or $S$ (i.e., the distribution of the sample is statistically close to the uniform distribution), respectively.
    \item[{\bf Superposition over $G$:}] There is a QPT algorithm that generates a state whose trace distance from $\ket{G}\coloneqq\sum_{g\in G}\ket{g}$ is $\negl(\secp)$. 
\end{itemize}
\end{definition}
\begin{remark}[A convention on ``Random sampling'' and ``Superposition over $G$'' properties]
In the rest of this paper, we assume that we can sample elements from \emph{exactly} uniform distributions of $G$ and $S$. Similarly, we assume that we can \emph{exactly} generate $\ket{G}$ in QPT. They are just for simplifying the presentations of our results, and all the results hold with the above imperfect version with additive negligible loss for security or correctness. 
\end{remark}

\begin{remark}
The above requirements are identical to those in \cite{TCC:JQSY19} except for the ``superposition over $G$'' property. 
We remark that all candidate constructions proposed in  \cite{TCC:JQSY19} satisfy this property. 
\end{remark}

We define one-wayness~\cite{TCC:JQSY19} as follows.
\begin{definition}[One-wayness]\label{def:one-way}
We say that a group action $(G,S,\star)$ with efficient algorithms is \emph{one-way} if for any non-uniform QPT adversary $\A$, we have 
\begin{align*}
    \Pr\left[
    g'\star s = g\star s
    :
    \begin{array}{l}
     s\gets S,
    g\gets G,
    g' \leftarrow \A(s,g\star s)
    \end{array}
    \right]=\negl(\secp).
\end{align*}
\end{definition}

\subsection{Swap-Trapdoor Functions Pairs}
\label{sec:STF}
In order to show the construction of two-tier quantum lightning with classical semi-verification from group actions, it is convenient to
use {\it swap-trapdoor function pairs} (STF), which was introduced in \cite{EC:HhaMorYam23}.
In this subsection, we define STF and point out that
STF can be constructed from group actions~\cite{TCC:JQSY19,EC:HhaMorYam23}.
Content of this subsection is taken from \cite{EC:HhaMorYam23}.

The formal definition of STF is given as follows.
\begin{definition}[Swap-Trapdoor Function Pair~\cite{EC:HhaMorYam23}]\label{def:STF}
A \emph{swap-trapdoor function pair (STF)} consists a set $(\setup,\eval,\swap)$ of algorithms such that
\begin{itemize}
\item
$\setup(1^\secp)\to(\pp,\td)$: It is a PPT algorithm that takes the security parameter $\secp$ as input, and outputs a public parameter $\pp$ and a trapdoor $\td$. The public parameter $\pp$ specifies functions $f_b^{(\pp)}:\calX\rightarrow \calY$ for each $b\in\bit$. We often omit the dependence on $\pp$ and simply write $f_b$ when it is clear from the context.  
\item
$\eval(\pp,b,x)\to y$: It is a deterministic classical polynomial-time algorithm that takes a public parameter $\pp$, a bit $b\in\bit$, and an element $x\in \calX$ as input, and outputs $y\in \calY$. 
\item
$\swap(\td,b,x)\to x'$: It is a deterministic classical polynomial-time algorithm that takes a trapdoor $\td$, a bit $b\in \bit$, and an element $x\in \calX$ as input, and outputs $x'\in \calX$.
\end{itemize}
We require a STF to satisfy the following:

\smallskip\noindent\textbf{Evaluation correctness.}
For any $(\pp,\td)\gets \setup(1^\secp)$ , $b\in\bit$, and $x\in \calX$, we have $\eval(\pp,b,x)=f_b(x)$.

\smallskip\noindent\textbf{Swapping correctness.}
For any $(\pp,\td)\gets \setup(1^\secp)$, $b\in \bit$, and $x\in \calX$, 
 if we let $x'\gets \swap(\td,b,x)$, then we have $f_{b\oplus 1}(x')=f_b(x)$ and $\swap(\td,b\oplus 1,x')=x$.
In particular, $\swap(\td,b,\cdot)$ induces an efficiently computable and invertible one-to-one mapping between $f^{-1}_0(y)$ and $f^{-1}_1(y)$ for any $y\in \calY$. 
 

\smallskip\noindent\textbf{Efficient random sampling over $\calX$.} 
There is a PPT algorithm that samples an almost uniform element of $\calX$ (i.e., the distribution of the sample is statistically close to the uniform distribution). 

\smallskip\noindent\textbf{Efficient superposition over $\calX$.} 
There is a QPT algorithm
 that generates a state whose trace distance from $\ket{\calX}=\frac{1}{\sqrt{|\calX|}}\sum_{x\in\calX}|x\rangle$ is $\negl(\secp)$.

\smallskip\noindent\textbf{Claw-freeness.}
For any non-uniform QPT algorithm $\A$, 
\begin{align*}
    \Pr[f_0(x_0)=f_1(x_1):(\pp,\td)\gets \setup(1^\secp),(x_0,x_1)\gets \A(\pp)]=\negl(\secp).
\end{align*}

 \end{definition}

\begin{remark}[A convention on ``Efficient random sampling over $\calX$'' and ``Efficient superposition over $\calX$'' properties]\label{rem:random_sampling_and_superposition}
In the rest of this paper, we assume that we can sample elements from \emph{exactly} the uniform distribution of $\calX$.  Similarly, we assume that we can \emph{exactly} generate $\ket{\calX}$ in QPT. They are just for simplifying the presentations of our results, and all the results hold with the above imperfect version with additive negligible loss for security or correctness. 
\end{remark}

A STF can be constructed from group actions. Let $(G,S,\star)$ be a group action with efficient algorithms (as defined in \Cref{def:GA_efficient}). Then, we construct a STF as follows.

\begin{itemize}
\item
$\setup(1^\secp)$: Generate $s_0\gets S$ and $g\gets G$, set $s_1\seteq g\star s_0$, and output $\pp\seteq (s_0,s_1)$ and $\td\seteq g$. 
For $b\in \bit$, we define $f_b:G\rightarrow S$ by 
    $f_b(h)\seteq h\star s_b$.
\item
$\eval(\pp=(s_0,s_1),b,h)$: Output $f_b(h)=h\star s_b$. 
\item
$\swap(\td=g,b,h)$:
If $b=0$, output $hg^{-1}$. If $b=1$, output $hg$. 
\end{itemize}

In \cite{EC:HhaMorYam23}, the following theorem is shown.
\begin{theorem}[\cite{EC:HhaMorYam23}]\label{thm:GA_to_STF}
     If $(G,S,\star)$ is one-way, then $(\setup,\eval,\swap)$ is claw-free. \label{item:OW_to_CF}
\end{theorem}

\subsection{Definition of Two-Tier Quantum Lightning with Classical Semi-Verification}
\label{sec:2QLclassical}
The formal definition of two-tier quantum lightning with classical semi-verification is given as follows.
The difference from the original two-tier quantum lightning \cite{TCC:KitNisYam21} is that the semi-verification algorithm $\SemiVer$ accepts a classical certificate $\cert$
instead of a quantum state, and there is an additional QPT algorithm $\Del$ that outputs $\cert$ on input $\psi$.
\begin{definition}[Two-Tier Quantum Lightning with Classical Semi-Verification]
Two-tier quantum lighting with classical semi-verification is a set $(\Setup,\StateGen,\Del,\SemiVer,\FullVer)$
of algorithms such that
\begin{itemize}
    \item 
    $\Setup(1^\secp)\to(\sk,\pk):$
    It is a QPT algorithm that, on input the security parameter $\secp$, outputs
    a classical secret key $\sk$ and a classical public key $\pk$.
    \item 
    $\StateGen(\pk)\to(\psi,\snum):$
    It is a QPT algorithm that, on input $\pk$, outputs 
    a quantum state $\psi$ and a classical serial number $\snum$.
    \item 
    $\Del(\psi)\to\cert:$
    It is a QPT algorithm that, on input $\psi$, outputs a classical certificate $\cert$.
    \item 
    $\SemiVer(\pk,\snum,\cert)\to\top/\bot:$
    It is a classical deterministic polynomial-time algorithm that, on input $\pk$, $\cert$, and $\snum$, outputs $\top/\bot$.
    \item 
    $\FullVer(\sk,\snum,\psi)\to\top/\bot:$
    It is a QPT algorithm that, on input $\sk$, $\psi$, and $\snum$, outputs $\top/\bot$.
\end{itemize}
We require the following properties.

\paragraph{Correctness:}
\begin{align}
\Pr\left[\top\gets\SemiVer(\pk,\snum,\cert):
\begin{array}{rr}
(\sk,\pk)\gets\Setup(1^\secp)\\
(\psi,\snum)\gets\StateGen(\pk)\\
\cert\gets\Del(\psi)
\end{array}
\right]\ge1-\negl(\secp)    
\end{align}
and
\begin{align}
\Pr\left[\top\gets\FullVer(\sk,\snum,\psi):
\begin{array}{rr}
(\sk,\pk)\gets\Setup(1^\secp)\\
(\psi,\snum)\gets\StateGen(\pk)\\
\end{array}
\right]\ge1-\negl(\secp).
\end{align}

\paragraph{Security:}
For any QPT adversary $\cA$,
\begin{align}
\Pr\left[
\top\gets\SemiVer(\pk,\snum,\cert)
\wedge
\top\gets\FullVer(\sk,\snum,\psi)
:
\begin{array}{rr}
(\sk,\pk)\gets\Setup(1^\secp)\\
(\cert,\psi,\snum)\gets\cA(\pk)\\
\end{array}
\right]\le\negl(\secp).    
\end{align}
\end{definition}

\subsection{Construction from Group Actions}
\label{sec:GAtoQL}
Here we construct two-tier quantum lightning with classical semi-verification from STF (and therefore from group actions).
\begin{theorem}
\label{thm:GAtoQL}
If group actions with the one-wayness property exist, then
two-tier quantum lightning with classical semi-verification exists.    
\end{theorem}

\begin{proof}[Proof of \cref{thm:GAtoQL}]
Let $(\setup,\eval,\swap)$ be a SFT.
From \cref{thm:GA_to_STF}, a SFT exists if
group actions that satisfy one-wayness exist.
From it, we construct two-tier quantum lightning with classical semi-verification $(\Setup,\StateGen,\Del,\SemiVer,\FullVer)$ as follows.
\begin{itemize}
    \item 
    $\Setup(1^\secp)\to(\sk,\pk):$
    Run $(\pp,\td)\gets\setup(1^\secp)$.
    Output $\pk\coloneqq\pp$ and $\sk\coloneqq(\pp,\td)$.
    \item 
    $\StateGen(\pk)\to(\psi,\snum):$
    Parse $\pk\coloneqq\pp$.
    For each $i\in[n]$, repeat the following:
    \begin{enumerate}
    \item
    Generate $\sum_{b\in\bit}\sum_{x\in\calX}\ket{b}\ket{x}\ket{\eval(\pp,b,x)}$.
    \item  
   Measure the third register in the computational basis
    to get the measurement result $y_i$.
    The post-measurement state is
    $\psi_i\coloneqq\ket{0}\otimes\sum_{x\in f_0^{-1}(y_i)}\ket{x}+\ket{1}\otimes\sum_{x\in f_1^{-1}(y_i)}\ket{x}$. 
    \end{enumerate}
    Output $\psi\coloneqq\bigotimes_{i\in[n]}\psi_i$ and $\snum\coloneqq\{y_i\}_{i\in[n]}$.
     \item 
    $\Del(\psi)\to\cert:$
    Parse $\psi=\bigotimes_{i\in[n]}\psi_i$.
    Measure $\psi_i$ in the computational basis to get the result $(b_i,z_i)\in\bit\times\calX$ for each $i\in[n]$.
    Output $\cert\coloneqq \{b_i,z_i\}_{i\in[n]}$.
    \item 
    $\SemiVer(\pk,\cert,\snum)\to\top/\bot:$
    Parse $\pk\coloneqq\pp$, $\cert=\{b_i,z_i\}_{i\in[n]}$ and $\snum=\{y_i\}_{i\in[n]}$.
    If $\eval(\pp,b_i,z_i)=y_i$ for all $i\in[n]$, output $\top$. Otherwise, output $\bot$.
    
    \item 
    $\FullVer(\sk,\psi,\snum)\to\top/\bot:$
    Parse $\sk\coloneqq(\pp,\td)$, $\psi=\bigotimes_{i\in[n]}\psi_i$, and $\snum=\{y_i\}_{i\in[n]}$.
    For each $i\in[n]$, do the following:
    \begin{enumerate}
        \item 
    Apply unitary $U$ such that $U\ket{b}\ket{x}\ket{0...0}=\ket{b}\ket{x}\ket{\eval(\pp,b,x)}$ on $\psi_i\otimes\ket{0...0}\bra{0...0}$, and measure
    the third register in the computational basis.
    If the measurement result is not $y_i$,  output $\bot$ and abort.
   \item 
    Apply the operation that maps $x$ to $\swap(\td,1,x)$ on the second register of the post-measurement state conditioned that the first register is 1.\footnote{Such an operation
    can be done in the following way: First, change $\ket{x}\ket{0...0}$ into
    $\ket{x}\ket{\swap(\td,1,x)}$. Then, change this state into 
    $\ket{x\oplus\swap(\td,0,\swap(\td,1,x))}\ket{\swap(\td,1,x)}=
    \ket{0...0}\ket{\swap(\td,1,x)}$.
    Finally, if we swap the first and the second register, we get
    $\ket{\swap(\td,1,x)}\ket{0...0}$.}
    Measure the first register in the Hadamard basis. If the result is 1 (i.e., if the result $\ket{-}$ is obtained), output $\bot$ and abort.
    \end{enumerate}
    Output $\top$.
\end{itemize}
The correctness for the semi-verification is clear. 
The correctness for the full-verification is also clear because
if each $\psi_i$ is
\begin{align}
\psi_i=\ket{0}\otimes\sum_{x\in f_0^{-1}(y_i)}\ket{x}
+\ket{1}\otimes\sum_{x\in f_1^{-1}(y_i)}\ket{x},
\end{align}
then the first operation of the semi-verification algorithm does not change it, and the second operation gives
\begin{align}
\ket{0}\otimes\sum_{x\in f_0^{-1}(y_i)}\ket{x}
+\ket{1}\otimes\sum_{x\in f_1^{-1}(y_i)}\ket{\swap(\td,1,x)}
&=
\ket{0}\otimes\sum_{x\in f_0^{-1}(y_i)}\ket{x}
+\ket{1}\otimes\sum_{x\in f_0^{-1}(y_i)}\ket{x}\\
&=
\ket{+}\otimes\sum_{x\in f_0^{-1}(y_i)}\ket{x}.
\end{align}

To show the security, we define the sequence of hybrids.
Hybrid 0 (\cref{hyb0}) is the original security game.
For the sake of contradiction, we
assume that there exists a QPT adversary $\cA$ such that
\begin{align}
\label{assumption}
\Pr[{\rm Hybrid~0}\to\top]\ge\frac{1}{\poly(\secp)}    
\end{align}
for infinitely many $\secp$.

\protocol{Hybrid 0}
{Hybrid 0}
{hyb0}
{
\begin{enumerate}
    \item 
    The challenger $\cC$ runs $(\td,\pp)\gets\setup(1^\secp)$,
    and sends $\pp$ to $\cA$.
    \item 
    $\cA$ sends a quantum state $\psi$ on the register $\regF$, a classical bit string $\cert=\{b_i,z_i\}_{i\in[n]}$,
    and 
    bit strings $\snum=\{y_i\}_{i\in[n]}$ to $\cC$.
    The register $\regF$ consists of sub-registers $(\regF_1,...,\regF_n)$.
    Each $\regF_i$ is a $(|\calX|+1)$-qubit register.
    \item 
    \label{SemiV}
    If $\eval(\pp,b_i,z_i)\neq y_i$ for an $i\in[n]$ , output $\bot$ and abort.
\item 
\label{FullV}
    From $i=1$ to $n$, $\cC$ does the following: 
    \begin{enumerate}
    \item 
    Apply $U$ such that $U\ket{b}\ket{x}\ket{0...0}=\ket{b}\ket{x}\ket{\eval(\pp,b,x)}$ on $\regF_i\otimes\ket{0...0}\bra{0...0}$,
    and measure the ancilla register in the computational basis. If the result is not $y_i$, output $\bot$ and abort.
        \item 
    Apply the operation that maps $x$ to $\swap(\td,1,x)$ on the second register of the post-measurement state on $\regF_i$ conditioned that the first qubit is 1.
    \label{cSWAP}
    \item
    Measure the first qubit of $\regF_i$ in the Hadamard basis. If the result is 1 (i.e., if the result is $\ket{-}$), output $\bot$ and abort.
    \end{enumerate}
    \item
    Output $\top$.
\end{enumerate}
}

\protocol{Hybrid 1}
{Hybrid 1}
{hyb1}
{
\begin{enumerate}
    \item
    The same as Hybrid 0. 
    \item
    The same as Hybrid 0.
    \item 
    The same as Hybrid 0.
 \item 
    From $i=1$ to $n$, $\cC$ does the following: 
    \begin{enumerate}
    \item 
    The same as Hybrid 0.
        \item 
        The same as Hybrid 0.
    \item
    Measure all qubits of $\regF_i$ in the computational basis to get the result $(e_i,w_i)\in\bit\times\calX$.
    \end{enumerate}   
    \item
    If there is $i\in[n]$ such that $b_i\neq e_i$, output $\top$. Otherwise, output $\bot$.
\end{enumerate}
}

\protocol{Hybrid 2}
{Hybrid 2}
{hyb2}
{
\begin{enumerate}
    \item
    The same as Hybrid 1. 
    \item
    The same as Hybrid 1.
    \item 
    The same as Hybrid 1.
\item 
The same as Hybrid 1 except that the step (b) is skipped for each
    $i\in[n]$.
    \item
    The same as Hybrid 1.
\end{enumerate}
}

\protocol{Hybrid 3}
{Hybrid 3}
{hyb3}
{
\begin{enumerate}
    \item
    The same as Hybrid 2. 
In addition, $\cC$ chooses $i^*\gets[n]$.
    \item
    The same as Hybrid 2.
    \item 
    The same as Hybrid 2.
    \item 
    The same as Hybrid 2.
\item 
    If $b_{i^*}\neq e_{i^*}$, output $\top$.
    Otherwise, output $\bot$.
\end{enumerate}
}

\begin{lemma}
\label{lem:01}
If 
$\Pr[{\rm Hybrid~0}\to\top]\ge\frac{1}{\poly(\secp)}$ for infinitely many $\secp$,
then
$\Pr[{\rm Hybrid~1}\to\top]\ge\frac{1}{\poly(\secp)}$ for infinitely many $\secp$.
\end{lemma}

\begin{proof}
Let $\Pr[b_1,...,b_n]$ be the probability that $\cA$ outputs $\{b_i\}_{i\in[n]}$ and
$\cC$ finishes \cref{cSWAP} of \cref{hyb0} without aborting.
Then
\begin{align}
  \epsilon&\coloneqq \Pr[\mbox{Hybrid 0}\to\top]\\
  &=\sum_{b_1,...,b_n}\Pr[b_1,...,b_n]
  \Tr[(\ket{+}\bra{+}^{\otimes n}\otimes I) \psi_{b_1,...,b_n}],
\end{align}
where $\psi_{b_1,...,b_n}$ is the state of the register $\regF$ after the end of \cref{cSWAP} of \cref{hyb0}.
From the standard average argument, we have
\begin{align}
\sum_{(b_1,...,b_n)\in G}\Pr[b_1,...,b_n]
\ge\frac{\epsilon}{2},    
\label{average}
\end{align}
where
\begin{align}
G\coloneqq\Big\{(b_1,...,b_n):
  \Tr[(\ket{+}\bra{+}^{\otimes n}\otimes I)\psi_{b_1,...,b_n}]\ge\frac{\epsilon}{2}
\Big\}.
\label{inG}
\end{align}
Define
\begin{align}
\Lambda_{b_1,...,b_n}    
\coloneqq
\sum_{(e_1,...,e_n)\in\bit^n\setminus (b_1,...,b_n)}
\ket{e_1,e_2,...,e_n}\bra{e_1,e_2,...,e_n}.
\end{align}
As is shown later, for any $m$-qubit state $\rho$, where $m\ge n$, and any $(b_1,...,b_n)\in\bit^n$, we have
 \begin{align}
 \label{PiLambda}
  \Tr[(\ket{+}\bra{+}^{\otimes n}\otimes I) \rho]
  \le
  2\Tr[(\Lambda_{b_1,...,b_n}\otimes I)\rho]+2^{-n+1}. 
\end{align}
Then
\begin{align}
\Pr[\mbox{Hybrid 1}\to\top]
&=    
\sum_{b_1,...,b_n}
\Pr[b_1,...,b_n]
\Tr[(\Lambda_{b_1,...,b_n}\otimes I)\psi_{b_1,...,b_n}]\\
&\ge
\sum_{(b_1,...,b_n)\in G}
\Pr[b_1,...,b_n]
\Tr[(\Lambda_{b_1,...,b_n}\otimes I)\psi_{b_1,...,b_n}]\\
&\ge
\sum_{(b_1,...,b_n)\in G}
\Pr[b_1,...,b_n]
\Big(\frac{1}{2}\Tr\left[(\ket{+}\bra{+}^{\otimes n}\otimes I)\psi_{b_1,...,b_n}\right]
-2^{-n}\Big)\label{Pinisuru}\\
&\ge\frac{\epsilon}{2}\Big(\frac{\epsilon}{4}-2^{-n}\Big)\label{epsilon}\\
&\ge\frac{1}{\poly(\secp)}.
\end{align}
Here, to derive \cref{Pinisuru}, we have used \cref{PiLambda}.
To derive \cref{epsilon}, we have used \cref{average} and \cref{inG}.
\end{proof}

\begin{lemma}
\label{lem:12}
$\Pr[{\rm Hybrid~2}\to\top]=\Pr[{\rm Hybrid~1}\to\top]$.
\end{lemma}
\begin{proof}
The controlled $\swap(\td,1,\cdot)$ operations do not change the distribution of $(e_1,...,e_n)$.
\end{proof}

\begin{lemma}
\label{lem:23}
$\Pr[{\rm Hybrid~3}\to\top]\ge\frac{1}{n}\Pr[{\rm Hybrid~2}\to\top]$.
\end{lemma}
\begin{proof}
It is trivial.
\end{proof}

In summary, from \cref{assumption} and \cref{lem:01,lem:12,lem:23}, we conclude that
$\Pr[{\rm Hybrid~3}\to\top]\ge\frac{1}{\poly(\secp)}$ for infinitely many $\secp$.
From such $\cA$ of Hybrid 3, we construct a QPT adversary $\cB$ that breaks the claw-freeness of STF as in \cref{cB}.
It is easy to verify that
\begin{align}
\Pr[\cB~ \mbox{wins}]&\ge 
\Pr[{\rm Hybrid~3}\to\top]\\
&\ge\frac{1}{\poly(\secp)}
\end{align}
for infinitely many $\secp$.
Therefore, $\cB$ breaks the claw-freeness of the STF.

\protocol{$\cB$}
{$\cB$}
{cB}
{
\begin{enumerate}
    \item 
    Get $\pp$ as input.
    \item 
    Run $(\psi,\{b_i,z_i\}_{i\in[n]},\{y_i\}_{i\in[n]})\gets\cA(\pp)$.
    Choose $i^*\gets[n]$.
    \item 
    Measure all qubits of $\regF_{i^*}$ in the computational basis to get the result $(e_{i^*},w_{i^*})$.
    \item 
    If $b_{i^*}=0$ and $e_{i^*}=1$, output $(z_{i^*},w_{i^*})$.
    If $b_{i^*}=1$ and $e_{i^*}=0$, output $(w_{i^*},z_{i^*})$.
    Otherwise, output $\bot$.
\end{enumerate}
}

\end{proof}

\begin{proof}[Proof of \cref{PiLambda}]
\label{App:PiLambda}
Because of the linearity of $\Tr$, we have only to show it for pure states.
We have
\begin{align}
\Tr[(\ket{+}\bra{+}^{\otimes n}\otimes I) \ket{\psi}\bra{\psi}]
&=
\Big\|(\bra{+}^{\otimes n}\otimes I)\ket{\psi}\Big\|^2\\
&=
\frac{1}{2^n}\Big\|\sum_{e_1,...,e_n}(\bra{e_1,...,e_n}\otimes I)|\psi\rangle\Big\|^2\\
&=
\frac{1}{2^n}\Big\|\sum_{(e_1,...,e_n)\neq(b_1,...,b_n)}(\bra{e_1,...,e_n}\otimes I)|\psi\rangle
+(\bra{b_1,...,b_n}\otimes I)|\psi\rangle
\Big\|^2\\
&\le
\frac{1}{2^{n-1}}\Big\|\sum_{(e_1,...,e_n)\neq(b_1,...,b_n)}(\bra{e_1,...,e_n}\otimes I)|\psi\rangle\Big\|^2\\
&~~~~+\frac{1}{2^{n-1}}\Big\|(\bra{b_1,...,b_n}\otimes I)|\psi\rangle\Big\|^2\\
&\le
\frac{1}{2^{n-1}}(2^n-1)\sum_{(e_1,...,e_n)\neq (b_1,...,b_n)} \|(\bra{e_1,...,e_n}\otimes I)|\psi\rangle\|^2
+2^{-n+1}\\
&=
\frac{2^n-1}{2^{n-1}}\Tr[(\Lambda_{b_1,...,b_n}\otimes I)\ket{\psi}\bra{\psi}]
+2^{-n+1}\\
&\le
2\Tr[(\Lambda_{b_1,...,b_n}\otimes I)\ket{\psi}\bra{\psi}]
+2^{-n+1}.
\end{align}
Here, we have used the triangle inequality and Jensen's inequality.
\end{proof}

\subsection{Construction of Digital Signatures with Quantum Revocable Signing Keys}
\label{sec:2QLtoDS}
We construct digital signatures with quantum revocable signing keys.
Its definition is given as follows.

\begin{definition}[Digital Signatures with quantum revocable signing keys]
The syntax, correctness, and security are the same as those of \cref{def:dswrsk}, except that
there is no deletion algorithm, $\Del$, and
$\cert$ is not a classical bit string, but the signing key $\sigk_i$ itself.
\end{definition}

Digital signatures with quantum revocable signing keys can be constructed from group actions as follows.
\begin{theorem}
If group actions with the one-wayness property exist,    
then digital signatures with quantum revocable signing keys exit.
\end{theorem}

\begin{proof}
Let $(\mathsf{QL}.\Setup,\mathsf{QL}.\StateGen,\mathsf{QL}.\Del,\mathsf{QL}.\SemiVer,\mathsf{QL}.\FullVer)$ be
two-tier quantum lightning with classical semi-verification,
which exists if 
group actions with the one-wayness property exist (\cref{thm:GAtoQL}).    
From it, we construct 
digital signatures with quantum revocable signing keys 
with the message space $\bit$ that satisfy one-time variants of correctness, EUF-CMA security, deletion correctness, and deletion security
as follows.
With similar proofs given in \cref{sec:skcert}, we can upgrade it to the full version of digital signatures
with quantum revocable signing keys.

\begin{itemize}
\item 
$\setup(1^\secp)\to(\ck,\pp):$
Run $(\sk,\pk)\gets\mathsf{QL}.\Setup(1^\secp)$.
Output $\ck\coloneqq \sk$ and $\pp\coloneqq\pk$.
    \item 
    $\KeyGen(\pp)\to(\sigk,\vk):$
    Parse $\pp=\pk$.
    Run $(\psi_0,\snum_0)\gets\mathsf{QL}.\StateGen(\pk)$ and
    $(\psi_1,\snum_1)\gets\mathsf{QL}.\StateGen(\pk)$.
    Output $\sigk\coloneqq(\psi_0,\psi_1)$ and $\vk\coloneqq(\snum_0,\snum_1)$.
    \item 
    $\Sign(\pp,\sigk, m)\to(\sigk',\sigma):$
    Parse $\pp=\pk$ and $\sigk=(\psi_0,\psi_1)$.
    Run
    $\mathsf{ql}.\cert\gets\mathsf{QL}.\Del(\psi_m)$.
    Output $\sigma\coloneqq\mathsf{ql}.\cert$ and $\sigk'\coloneqq(1-m,\psi_{1-m})$.
    \item 
    $\Ver(\pp,\vk,m,\sigma)\to\top/\bot:$
    Parse $\pp=\pk$, $\vk=(\snum_0,\snum_1)$, and $\sigma=\mathsf{ql}.\cert$.
    Run $\SemiVer(\pk,\snum_m,\mathsf{ql}.\cert)$ and output its output.
    \item 
    $\Cert(\pp,\vk,\ck,\sigk,S)\to\top/\bot:$
    Parse $\pp=\pk$, $\vk=(\snum_0,\snum_1)$, $\ck=\sk$.
    \begin{itemize}
    \item If $S=\emptyset$, parse $\sigk=(0,1,\psi_0,\psi_1)$. (Output $\bot$ if $\sigk$ is not of this form.)  
    Output $\top$ if $\mathsf{QL}.\FullVer(\sk,\snum_b,\psi_b)=\top$ for both $b=0,1$. 
    \item If $S=\{m\}$ for some $m\in \bit$, 
    parse $\sigk=(1-m,\psi_{1-m})$. (Output $\bot$ if $\sigk$ is not of this form.)
    Output $\top$ if $\mathsf{QL}.\FullVer(\sk,\snum_{1-m},\psi_{1-m})=\top$.
    \item If $S=\{0,1\}$, output $\bot$.  
    \end{itemize}
\end{itemize}
One-time correctness, one-time EUF-CMA security, one-time deletion correctness, and one-time deletion security
are shown in similar ways as those in \cref{sec:onetime}.
\end{proof} 
\fi

\end{document}